\documentclass[a4paper, intlimits, 10pt]{amsart}

\usepackage[T1]{fontenc}
\usepackage{float}

\usepackage{times}
\usepackage[english]{babel}
\usepackage[T1]{fontenc}
\usepackage{enumitem}
\usepackage{bm}
\usepackage{graphicx}

\usepackage{subcaption}
\usepackage{booktabs}
\usepackage{multirow}
\usepackage{tabularx}
\usepackage{color}
\usepackage{eurosym}
\usepackage{url}
\usepackage[numbers,sort&compress]{natbib}
\usepackage{arydshln}

\usepackage{comment}

\usepackage{amsthm, amssymb, amsfonts}

\usepackage{accents}

\newtheorem{theorem}{Theorem}[section]
\newtheorem{corollary}[theorem]{Corollary}

\newtheorem{lemma}[theorem]{Lemma}
\newtheorem{proposition}[theorem]{Proposition}
\newtheorem{definition}[theorem]{Definition}

\newtheorem{remark}[theorem]{Remark}

\numberwithin{equation}{section}

\makeatletter

\newcommand{\raisedrule}[3][0em]{\leavevmode\leaders\hbox{\rule[#1]{1pt}{#2}}\hfill\kern0pt \;\;\small\textsc{#3}\;\; \leavevmode\leaders\hbox{\rule[#1]{1pt}{#2}}\hfill\kern0pt}

\def\@pauthors{}
\newcounter{pauthorcnt}
\newcommand{\pauthor}[2][]{%
  \stepcounter{pauthorcnt}%
  \ifnum\value{pauthorcnt}=1
    \g@addto@macro\@pauthors{#2%
      \if\relax\detokenize{#1}\relax\else\unskip\textsuperscript{#1}\fi}%
  \else
    \g@addto@macro\@pauthors{,\ #2%
      \if\relax\detokenize{#1}\relax\else\unskip\textsuperscript{#1}\fi}%
  \fi
}

\def\@paffils{}
\newcounter{paffilcnt}
\renewcommand\thepaffilcnt{\alph{paffilcnt}}

\newcommand{\paffil}[1]{%
  \stepcounter{paffilcnt}%
  \edef\@tmpaffil{\noexpand\textsuperscript{\thepaffilcnt}\noexpand\;#1\noexpand\par\noexpand\vspace{2pt}}%
  \expandafter\g@addto@macro\expandafter\@paffils\expandafter{\@tmpaffil}%
}

\def\@pemails{}
\newcounter{pemailcnt}
\renewcommand\thepemailcnt{\arabic{pemailcnt}}

\newcommand{\pemail}[1]{%
  \stepcounter{pemailcnt}%
  \edef\@tmpemail{\noexpand\textsuperscript{\thepemailcnt}\noexpand\,\noexpand\mbox{\noexpand\ttfamily\noexpand\scriptsize #1}\noexpand\par}%
  \expandafter\g@addto@macro\expandafter\@pemails\expandafter{\@tmpemail}%
}

\def\@pabstract{}
\newcommand{\pabstract}[1]{\gdef\@pabstract{#1}}

\def\@pkeywords{}
\newcommand{\pkeywords}[1]{\gdef\@pkeywords{#1}}

\def\@ppaperinfo{}
\newcommand{\ppaperinfo}[1]{\gdef\@ppaperinfo{#1}}

\newcommand{\printtitlepage}{%
  \vspace{-1em}%
  \begin{center}{\large\@pauthors\par}\end{center}%
  \vspace{6pt}%
  \noindent
  \begin{tabular}[t]{@{}p{0.62\textwidth}@{\hspace{0.04\textwidth}}p{0.34\textwidth}@{}}
    \vspace{0pt}
    \raisedrule[0.2em]{0.1pt}{\footnotesize Abstract}
    \par\vspace{3pt}
    \footnotesize\@pabstract
    \par\vspace{6pt}
    \raisedrule[0.2em]{0.1pt}{\footnotesize Keywords}
    \par\vspace{3pt}
    \footnotesize\@pkeywords
  &
    \vspace{0pt}
    \raisedrule[0.2em]{0.1pt}{\footnotesize Authors Info}
    \par\vspace{3pt}
    \scriptsize\itshape\@paffils
    \scriptsize\@pemails
    \vspace{6pt}
    \raisedrule[0.2em]{0.1pt}{\footnotesize Paper Info}
    \par\vspace{3pt}
    \scriptsize\itshape\@ppaperinfo
  \end{tabular}
  \par\vspace{6pt}%
  \noindent\rule{\textwidth}{0.3pt}\par
  \vspace{10pt}%
}
\makeatother

\begin{document}

\pauthor[a,b,1]{Samuel Drapeau}
\pauthor[a,2]{Peng Luo}
\pauthor[a,3]{Xuan Tao}
\pauthor[b,4]{Tan Wang}

\paffil{School of Mathematical Sciences, Shanghai Jiao Tong University, Shanghai, China}
\paffil{Shanghai Advanced Institute of Finance, Shanghai Jiao Tong University, Shanghai, China}

\pemail{sdrapeau@saif.sjtu.edu.cn}
\pemail{peng.luo@sjtu.edu.cn}
\pemail{taoxuan@htsc.com.cn}
\pemail{tan.wang@saif.sjtu.edu.cn}

\ppaperinfo{\textnormal{Last updated:} \today \newline  \textnormal{MSC:} 91B50, 60H10, 91G20, 91B24.\newline \textnormal{JEL:} F31, G12, G15.}

\title{One Currency, Two Forward Prices: The Onshore-Offshore Renminbi Puzzle}
\author{}
\date{\today}
\maketitle

\pabstract{%
	Partially convertible economies face a market-design problem: trade integration, cross-border investment, and domestic balance-sheet exposure increase the demand for currency hedging before full financial integration is complete.
	China adopted a distinctive architecture for this problem by fostering a deliverable offshore Renminbi market (CNH) alongside the segmented onshore market (CNY), rather than relying only on non-deliverable forwards.
	This creates two venues for closely related claims on the same currency.
	Spot prices are tightly linked, yet CNY and CNH forwards display a persistent and economically large discrepancy.
	We study that discrepancy in a joint equilibrium model for spot and forward trading with transaction costs and segmented supply.
	In the benchmark case with common constant supply and deterministic costs, spot parity implies a forward differential with the wrong sign relative to the data.
	Random offshore stress, modeled as a jump in trading costs, overturns this benchmark while preserving tight spot parity.
	The model yields a semi-explicit representation in the CNY/CNH application and a calibration of the observed forward discrepancy in terms of the market-implied likelihood and severity of offshore liquidity stress.}

\pkeywords{On/Off-Shore Forward Pricing, Spot/Forward Equilibrium, Transaction Costs, FBSDE, Market Segmentation, Liquidity Stress}

\printtitlepage

\section{Introduction}
Developing economies that retain partial capital-account restrictions face a persistent market design problem.
Convertibility limits can reduce exposure to abrupt capital-flow reversals and give domestic financial markets time to deepen.
At the same time, trade integration, outward and inward investment, and larger cross-border balance sheets create a growing demand for currency hedging.
That demand comes not only from international investors, but also from domestic exporters, importers, and financial institutions whose cash flows are tied to foreign currency markets.
The issue is therefore not a binary choice between control and openness.
It is a sequencing problem: how to expand hedging access and market depth while maintaining a gradual integration of domestic and offshore financial conditions.

Offshore non-deliverable forward (NDF) markets have been one important response to this problem.
They allow exchange-rate risk to be traded without full access to the domestic currency market, and they often become the first liquid hedging venue for a partially convertible currency.
But NDFs are indirect claims: they settle in a reference currency rather than delivering the underlying domestic currency.
As they grow, they can also become central to price discovery, expectations, and offshore funding conditions.
Many economies have therefore used a range of approaches, from widening participation in NDF markets to expanding onshore hedging access and developing market-connect programs.
China's response to the same broad problem was distinctive.
Although a CNY NDF market existed, beginning in 2009 China promoted a deliverable offshore Renminbi market in Hong Kong (CNH) while the onshore Renminbi market (CNY) remained segmented.\footnote{Over the past decade and a half, China has implemented additional policies to liberalize RMB trade.
	Notable measures include enabling RMB trade settlement in over 25 countries via local banks, launching stock (2015) and bond (2017) interconnect programs for qualified financial institutions to access onshore markets, and expanding the use of swap lines with foreign central banks.
	These developments have contributed to a decline in the relative turnover of offshore RMB (CNH) compared to onshore RMB (CNY).}
At the same time, the traditional CNY non-deliverable forward market has gradually been replaced by the deliverable offshore CNH market.\footnote{According to the BIS Triennial Survey, the daily turnover of CNY NDFs decreased by 40\% from \$17.1 billion in 2013 to \$10.4 billion in 2016.
	During the same period, the daily turnover of CNH forwards doubled to \$16.4 billion, surpassing NDFs.}
This architecture is not simply an intermediate step between restrictions and liberalization.
It creates two deliverable venues for closely related claims on the same currency: an onshore market embedded in a more regulated financial system and an offshore market operating under more flexible but potentially more fragile liquidity conditions.
One offshore RMB is still one Renminbi, which gives spot prices a strong common anchor.
Yet the two venues differ in their hedging technology, funding conditions, and exposure to liquidity stress.
The CNY/CNH pair is therefore a natural laboratory for a broader question: how do spot and forward prices behave when the same terminal payoff is traded across segmented venues with different dynamic hedging costs?

Empirically, the spot relation is extremely tight.
Aside from occasional deviations during periods of financial turbulence,\footnote{For example, during and after the 2015 financial crisis in mainland China.} the spot price difference between the two markets rarely exceeds a few basis points, see Figure \ref{fig:spot_ratio}.
Over the period 2010--2021, the log ratio of CNY/CNH spot prices averaged $-0.01\%$, with interquartile ranges of $-0.1\%$ and $0.1\%$.
If spot parity were the dominant force at all horizons, forward prices should also be close across markets.
Yet the forward market shows a markedly different picture.
Forward price deviations are systematically negative and economically large across maturities.
More precisely,
\begin{equation*}
	r^{Y/H}_T := \frac{1}{T} \ln\left(\frac{F^Y_T}{F^H_T}\right) \approx -4.4\%
\end{equation*}
consistently since the establishment of the CNH market and across all maturities $T = 1M, 2M, 3M, 6M, 1Y$, see Figure \ref{fig:nforward_box}, and Table \ref{table:forward_ratio}.

\begin{figure}[H]
	\centering
	\includegraphics[width=\textwidth]{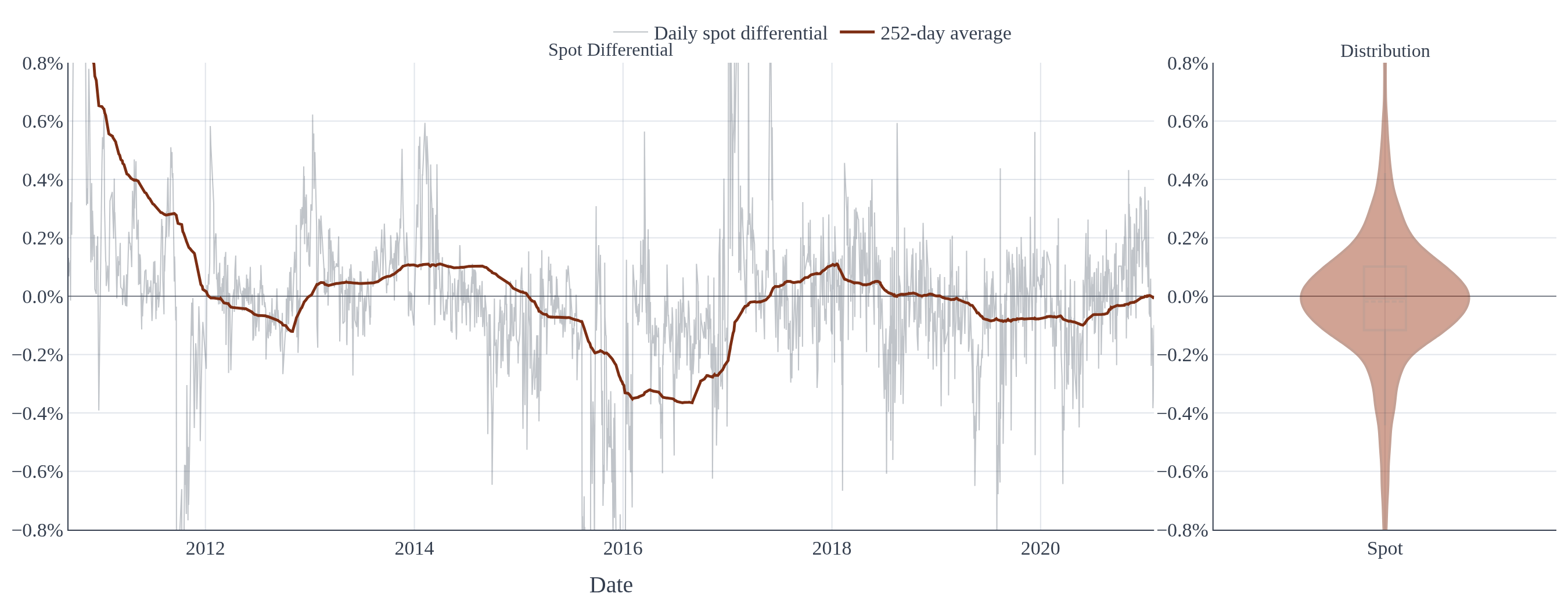}
	\caption{Log ratio of CNY/CNH spot rates in $\%$.}
	\label{fig:spot_ratio}
\end{figure}

\begin{figure}[h!]
	\centering
	\includegraphics[width=\textwidth]{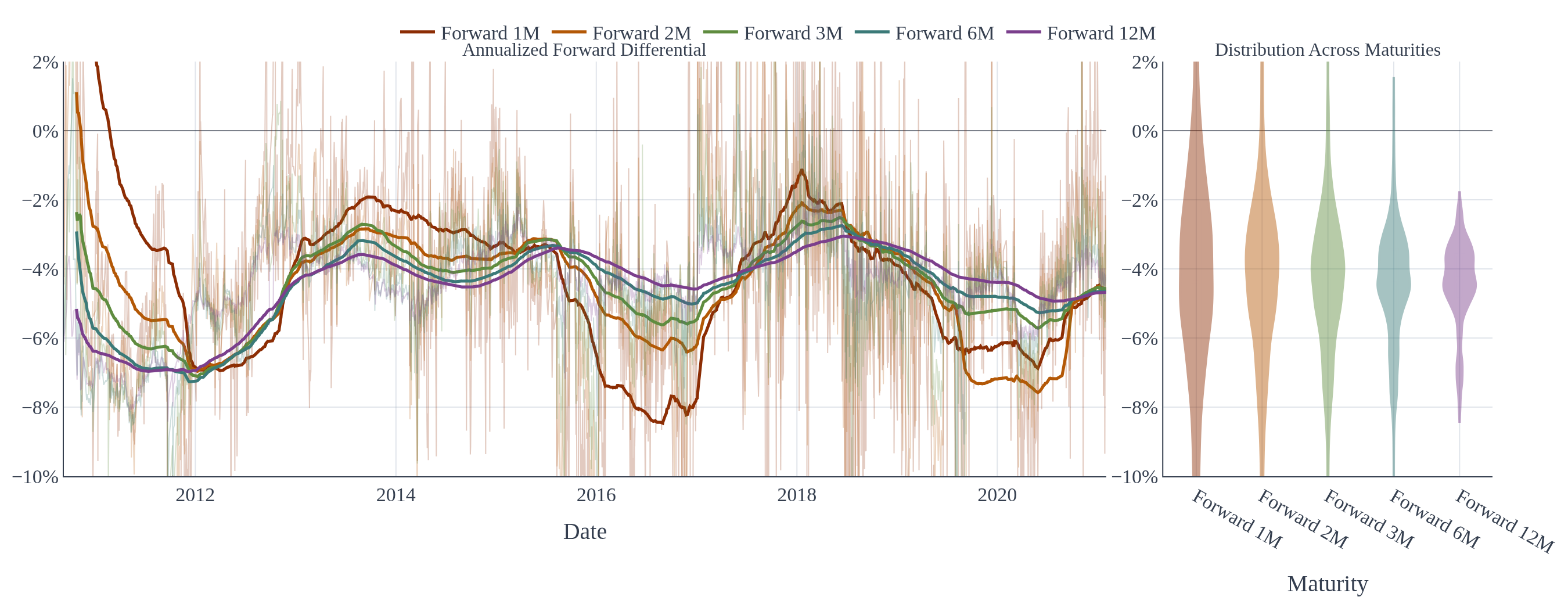}
	\caption{Box plot of the annualized forward differential $r^{Y/H}_T = \frac{1}{T}\ln(F^Y_T/F^H_T)$ across maturities.}
	\label{fig:nforward_box}
\end{figure}

\begin{table}[H]
	\begin{center}
		\begin{tabular}{@{}lcrrrrr@{}}
			\toprule
			           &  & $r^{Y/H}_1$ & $r^{Y/H}_2$ & $r^{Y/H}_3$ & $r^{Y/H}_6$ & $r^{Y/H}_{12}$ \\
			\midrule
			mean       &  & -4.47\%     & -4.58\%     & -4.41\%     & -4.44\%     & -4.39\%        \\
			std        &  & 4.98\%      & 3.36\%      & 2.07\%      & 1.54\%      & 1.17\%         \\
			$q_{25\%}$ &  & -6.36\%     & -5.71\%     & -5.50\%     & -5.20\%     & -4.86\%        \\
			Median     &  & -4.29\%     & -4.18\%     & -4.22\%     & -4.32\%     & -4.26\%        \\
			$q_{75\%}$ &  & -2.27\%     & -2.87\%     & -3.14\%     & -3.43\%     & -3.59\%        \\
			\bottomrule
		\end{tabular}
		\caption{Summary statistics about the yield $r^{Y/H}_T = \frac{1}{T}\ln(F^{Y}_T/F^{H}_T)$ from 2010 to 2021 for maturities $1$, $2$, $3$, $6$ and $12$ months.
		}
		\label{table:forward_ratio}
	\end{center}
\end{table}

The empirical object is therefore a persistent and economically large difference between onshore and offshore forward prices, despite nearly identical spot prices across the two venues.

A natural first benchmark is the standard no-arbitrage relation between spot prices, forward prices, and funding rates.
In stylized classical foreign exchange theory, this relation gives
\begin{equation*}
	F^Y_T = S^Y e^{(r^Y - r)T} \quad \text{and} \quad F^H_T = S^H e^{(r^H - r)T}
\end{equation*}
where $r^Y$, $r^H$, and $r$ are the funding rates onshore, offshore, and foreign (USA), respectively.
Since $S^Y \approx S^H$, we would have
\begin{equation*}
	r^{Y/H}_T = r^Y - r^H
\end{equation*}
At the macro level, however, the relevant onshore-offshore funding differential is much smaller on average than the forward wedge and does not display the same persistent pattern.
What it does exhibit are rare but extraordinary spikes, especially around the 2016 CNH stress episode, when offshore liquidity conditions tightened abruptly.
This suggests that the forward wedge is not explained by a stable carry differential alone.
It reflects a state-contingent hedging environment: offshore trading may be cheaper in normal times but much more expensive in stress states.

This points to a different mechanism.
Because the two venues ultimately trade the same currency, spot parity can remain tightly anchored by convertibility channels, institutional linkages, and arbitrage.
Forward pricing is different: dealers who absorb forward demand must hedge dynamically in the spot market, so the relevant marginal object is the expected future cost of hedging rather than the current spot differential alone.
If the offshore market is cheap in normal times but exposed to rare and severe liquidity squeezes, those state-contingent hedging costs are capitalized into the offshore forward price even when the spot differential remains negligible.
The spikes observed in offshore funding markets, together with contemporaneous increases in margin requirements, point directly to this type of random cost shock.
Under this interpretation, tight spot parity and a persistently more expensive offshore forward are not contradictory: they indicate that state-contingent offshore frictions are priced mainly in forwards rather than in spot.

\begin{figure}[h!]
	\centering
	\includegraphics[width=\textwidth]{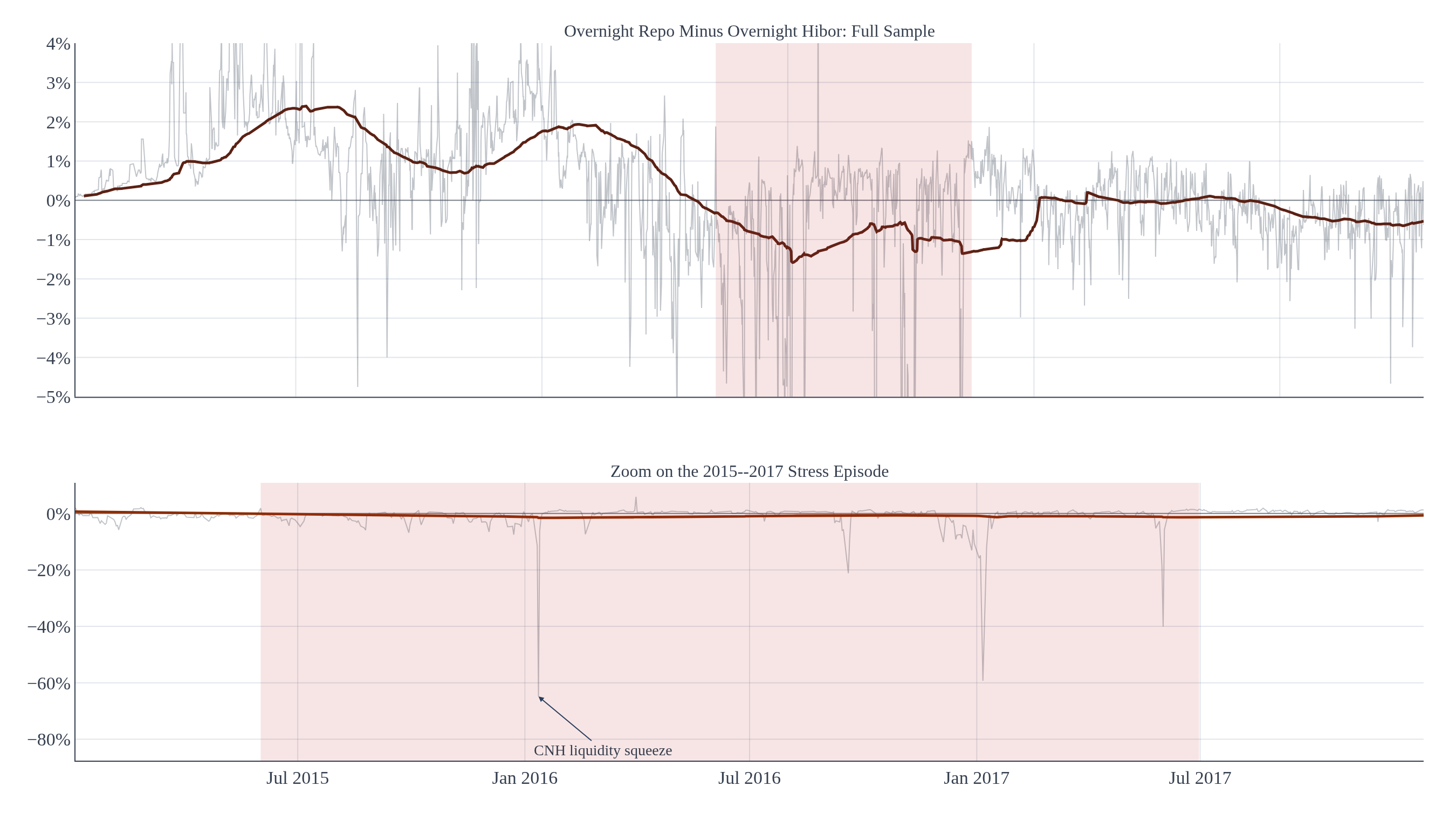}
	\caption{Onshore overnight repo minus offshore overnight Hibor in $\%$. The top panel focuses on the regular range, while the bottom panel shows the full stress episode. The differential is moderate on average but exhibits rare and pronounced spikes, especially around the 2016 CNH liquidity squeeze.}
	\label{fig:hibor_diff}
\end{figure}

This is the mechanism we formalize in the model.
The model allows both supply conditions and trading costs to differ across venues, but in the CNY/CNH application the central force is the possibility of abrupt offshore cost deterioration.
Onshore trading is more regulated and typically more expensive at the margin, while offshore trading is more flexible and cheaper in normal times but can become suddenly much more expensive when liquidity conditions tighten.
The January 2016 CNH squeeze is the natural example: the overnight funding rate jumped from roughly $3\%$ to above $70\%$ within days, together with a sharp increase in margin requirements.

Our contribution is to study this mechanism in a joint equilibrium model for spot and forwards.
The spot market is dynamic, the forward market is static, and the two are linked because forward dealers hedge through future spot trading.
This mixed static-dynamic structure separates current spot linkage from future hedging costs.
The deterministic benchmark gives a sharp empirical test: under common constant supply and spot parity, it implies $F^Y-F^H\approx m(c^Y-c^H)$, so with the normal ordering $c^Y>c^H$ it predicts $F^Y>F^H$, the opposite of the data.
Random offshore stress is therefore not an auxiliary feature but the mechanism that reverses the sign while preserving tight spot parity.
In the jump specification, the ordering is $\underline{c}^H<c^Y<\overline{c}^H$: offshore hedging is cheaper in normal times and more expensive after a stress event.
The observed forward wedge can then be read as a market-implied assessment of abrupt offshore stress: prices internalize both the perceived probability of a future liquidity squeeze and the expected severity of the associated hedging cost increase.

\subsection{Literature Review}

The paper is related to three strands of literature.

The first strand concerns the institutional development of offshore RMB markets and the empirical relation between CNY and CNH.
The emergence of the offshore RMB market and its interaction with the broader Asian offshore-currency architecture are documented, for instance, by \citet{maziad2012rmb}, \citet{craig2013development}, \citet{imf2020}, and \citet{cheung2021evolution}.
The shift from non-deliverable to deliverable offshore RMB instruments is discussed by \citet{mccauley2016ndf} and \citet{ehlers2016changing}, while \citet{packer2019renminbi} document the later tilt of renminbi turnover back toward onshore trading.
Empirical work on CNY/CNH itself has mostly focused on price discovery, volatility spillovers, and liquidity or policy determinants of the onshore-offshore gap; see, among others, \citet{shen2014}, \citet{wu2013dynamic}, \citet{cheung2014offshore}, \citet{Leung2014interactions}, \citet{yan2015}, \citet{shu2015one}, and \citet{funke2015assessing}.
That literature is informative about co-movements, market segmentation, and RMB internationalization, but it does not provide an equilibrium explanation for a persistent deliverable forward wedge under near-spot parity.

The second strand studies exchange rates and asset prices under segmentation, liquidity frictions, or capital-account distortions.
In particular, \citet{gabaix2015international} emphasize the role of financial intermediaries and international liquidity in exchange-rate determination.
The paper is also connected to the literature on covered interest parity failures and cross-currency bases, where balance-sheet costs and hedging imbalances generate persistent deviations from frictionless no-arbitrage benchmarks; see \citet{du2018deviations} and \citet{sushko2016failure}.
For the RMB specifically, \citet{jermann2022two} and \citet{mccauley2018recent} study the exchange-rate-policy regime and its relation to broader currency co-movements.
Our setting is narrower and more microstructural: the two venues trade the same RMB payoff, the spot side is tightly linked across venues, and the object of interest is the forward wedge generated by heterogeneous hedging conditions rather than the level of the exchange rate itself.

The third strand is the equilibrium literature with transaction costs and trading frictions.
Dynamic asset-pricing effects of trading costs appear already in \citet{vayanos1999equilibrium}, \citet{lo2004asset}, and \citet{buss2019dynamic}.
More recent equilibrium work with transaction costs includes \citet{bouchard2018equilibrium}, \citet{herdegen2021equilibrium}, \citet{weston2018existence}, \citet{escauriaza2022radner}, and \citet{gonon2021asset}.
Related optimal trading papers such as \citet{garleanu2016dynamic} and \citet{almgren2016option} study super-linear trading frictions in dynamic hedging problems, but with prices taken as exogenous.
Relative to these literatures, the present paper combines endogenous spot and forward pricing, random trading frictions, and a jump specification motivated by offshore RMB stress episodes.
This is the feature that allows the model to address both the wrong-sign deterministic benchmark and the calibration of market-implied stress beliefs.

\subsection{Economic Mechanism}
Section 3 specializes the theory to the CNY/CNH case.
At the economic level, the mechanism is simple.
We consider two trading venues, onshore $Y$ and offshore $H$, for the same RMB payoff $\mathfrak{G}$ delivered at time $T$.
On each venue, the spot price is determined dynamically while the forward price is determined in static equilibrium.
The two markets differ only through local trading frictions and local balance-sheet capacity.

More precisely, on each venue $i\in\{Y,H\}$, the spot price $S^i$ of the common payoff satisfies
\begin{equation*}
	S^i(t) = S^i(0) + \int_0^t \mu^i(s)\,ds + \int_0^t \sigma^i(s)\,dW(s),
	\qquad S^i(T)=\mathfrak{G},
\end{equation*}
while the forward price $F^i$ is the static equilibrium price at which dealer supply meets investor demand.
The link between the two comes from hedging: dealers absorb forward demand and then replicate that exposure by trading dynamically in spot.
Trading $q$ units on the spot market entails a quadratic execution cost $\frac{c^i(t)}{2}q^2$, so the forward price depends on the entire future path of spot hedging costs, not only on current funding conditions.
At the same time, the local availability of the asset is described by a supply process
\begin{equation*}
	M^i(t) = M^i(0) + \int_0^t m^i(s)\,ds,
\end{equation*}
which determines how much inventory the arbitrage sector can absorb on each venue.
For CNY/CNH, the key asymmetry is in trading costs.
Onshore costs are structurally higher and more stable because of capital-account restrictions, whereas offshore costs are lower in normal times but can jump sharply when liquidity conditions tighten.
The important point is that these changes need not build up smoothly over calendar time: they can be triggered by a liquidity event that arrives unexpectedly, and from that random date onward hedging costs can change dramatically.
This is exactly the type of specification studied in the application section: if $\tau$ denotes a random stress time, then the offshore cost process can be modeled by
\begin{equation*}
	c^H(t)=
	\begin{cases}
		\underline{c}^H, & t<\tau,     \\
		\overline{c}^H,  & t\geq \tau,
	\end{cases}
\end{equation*}
with $\underline{c}^H \ll c^Y < \overline{c}^H$. The general model also allows time-varying supply, but in the baseline calibration we keep the supply rate constant and common across venues and let the random shift operate through costs, because the data point most directly to sudden changes in funding and margin conditions.

Three groups of agents interact on each venue.
Investors demand forward contracts for hedging purposes, with demand decreasing in the forward price.
Dealers absorb that forward demand and hedge it by trading in the spot market.
Their spot inventory evolves according to their trading rate, and they face a trade-off between expected trading gains, execution costs, inventory risk, and a terminal mismatch penalty if their final inventory does not match their forward commitment.
Finally, arbitrageurs trade in the spot market and absorb part of the available supply; in equilibrium, they expand their position until marginal trading gains equal marginal trading costs.

This leads to a mixed static-dynamic equilibrium.
For a given forward supply, spot-market clearing determines the spot dynamics and the dealer inventory path.
Conversely, the dealer's willingness to supply forwards depends on the dynamic spot hedging problem they face, because a forward position is valuable only through the future sequence of spot trades required to hedge it.
Equating dealer supply with investor demand then pins down the forward price.
Hence the forward is not priced in a separate static block: it is linked to the entire dynamic spot equilibrium through expected future inventory and trading costs.
This feedback from future hedging conditions to current forward supply is the key economic channel of the paper.

In the CNY/CNH application, the mechanism gives a clear interpretation of the data.
Spot parity can remain tight because both venues ultimately trade the same currency and are linked by policy and arbitrage.
At the same time, forward prices may diverge because the expected cost of hedging forward positions differs across venues.
In particular, if the offshore market is cheap in normal times but exposed to rare and severe liquidity squeezes, those state-contingent costs are capitalized into the offshore forward price even when the spot differential remains small.
The role of randomness is substantive rather than cosmetic. In the constant-cost benchmark, once the spot equality constraint $S^Y(0)=S^H(0)$ is imposed, the model implies
\begin{equation*}
	\Delta F = m(c^Y-c^H),
\end{equation*}
With common constant supply $m$ on both sides and the empirically relevant ordering $c^Y > c^H$, this gives the wrong sign for the forward wedge: it predicts that the onshore forward should be more expensive than the offshore one. The main empirical message of the application is that random offshore stress reverses this deterministic prediction while preserving tight spot parity. This is also why we do not rely here on cross-market supply differences to generate the wedge. In the jump specification, the offshore forward price already reflects the possibility of a future stress episode, so current forward prices incorporate expected future hedging conditions rather than only current trading conditions.
The numerical exercise then gives a simple interpretation of the wedge. Holding fixed the common supply rate and the normal-times offshore cost, each target value of $F^Y-F^H$ can be mapped, under spot parity, into a market-implied pair $(\lambda,\overline c^H)$ describing the probability of an offshore stress episode and the size of the offshore cost increase. The forward wedge therefore tells us which combinations of stress probability and stress severity are consistent with the observed data.

\subsection{Outline}

The paper is organized as follows.
Section 2 presents the general spot-forward equilibrium model, states the existence results, and isolates the no-risk-aversion representation used in the application.
Section 3 applies this representation to CNY/CNH: it formulates the spot-parity restriction, shows the deterministic sign failure, introduces random offshore stress, and interprets the forward wedge through an inverse calibration of stress likelihood and severity.
Section 4 contains the proofs of the equilibrium results, the derivations behind the CNY/CNH formulas stated in Section 3, and the technical estimates for the small-risk-aversion argument.

\section{Spot and Forward Market Equilibrium with Random Transaction Costs and Supply}

We now introduce the general equilibrium model underlying the economic mechanism described in the introduction.
Throughout, we work on a filtered probability space $(\Omega, \mathcal{F}, \mathbb{F}=(\mathcal{F}_t)_{0\leq t\leq T}, \mathbb{P} )$ generated by a Brownian motion $W$ and an independent compensated Poisson random measure $\tilde{N}$ defined on $\Omega \times [0,T] \times \mathbb{R}$ with compensator $\nu(t,dx)dt$.\footnote{Intuitively, but not necessarily, the Brownian motion drives idiosyncratic risk, while the Poisson process represents exogenous sudden changes in market constraints.}
For notational convenience, we denote by $X = (W, \tilde{N})$ the resulting martingale.
We use the following spaces:
\begin{itemize}[fullwidth]
	\item $\mathbb{S}^\infty$ and $\mathbb{S}^p$ denote the set of c\`adl\`ag processes $\mu$ such that $\sup_{t\leq T}|\mu(t)|$ is either bounded or $p$-integrable, respectively.
	\item $\mathbb{H}$ denote those processes $\sigma=(\sigma_1, \sigma_2)$ for which $\int \sigma dX$ is well defined.\footnote{That is $\sigma_1$ is a progressively measurable stochastic process, while $\sigma_2$ is such that $\int_{\mathbb{R}}\sigma_2 \nu(dx)$ is progressively measurable.}
	\item $\mathbb{H}^p$ and $\mathbb{H}_{BMO}$ are those integrands $\sigma$ in $\mathbb{H}$ such that $\int \sigma dX$ is in $\mathbb{S}^p$ or is $BMO$, respectively.
\end{itemize}
For $\sigma$ in $\mathbb{H}^2$, we define $\sigma^2 := \sigma_1^2 + \int_{\mathbb{R}}\sigma_2^2 d\nu$, since this is the quadratic variation density of the process $\int \sigma dX$.

We consider a fixed bounded measurable contingent claim $\mathfrak{G}$.
This contingent claim is dynamically traded on the spot market at price $S=(S(t))$ and statically on the forward market at price $F$.
In equilibrium, the spot price follows a Bachelier dynamic
\begin{equation*}
	S(t) = S(0) + \int_0^t \mu ds + \int_0^t \sigma dX, \quad \text{with } S(T) = \mathfrak{G}
\end{equation*}
where $\mu$ is in $\mathbb{S}^\infty$ and $\sigma$ is in $\mathbb{H}_{BMO}$.
Market constraints enter through supply and trading costs. The supply of the underlying asset is denoted by $M(t) = M(0) + \int_0^t m\,ds$, where $m$ is the supply rate, and spot trading costs are quadratic in quantity: trading an amount $q$ at time $t$ costs $c(t)\frac{q^2}{2}$.

There are three types of market participants:

\begin{itemize}[fullwidth]
	\item \textbf{Investors}, with a forward demand function $F \mapsto \mathfrak{d}(F)$ that is strictly decreasing and continuous in the forward price $F$.
	\item \textbf{Retailers} who, after selling an amount $\mathfrak{s}$ of forward contracts at price $F$, hedge in the spot market with a trading rate $q$ in $\mathbb{H}^2$ and corresponding inventory $Q = \int q\,ds$.
	      Their objective function is
	      \begin{align*}
		      J(q, \mathfrak{s}) = & E\Bigg[ \int_{0}^{T}Q \mu dt + \mathfrak{s}(F-S(T)) - \frac{\rho}{2}\left(Q(T)-\mathfrak{s}\right)^2\Bigg] \\
		                           & \quad - \frac{1}{2}E\Bigg[\int_{0}^{T} c q^2dt + \phi \int_{0}^{T} Q^2 \sigma^2  dt  \Bigg]
	      \end{align*}
	      where $\phi$ is the risk-aversion parameter and $\rho(Q(T)-\mathfrak{s})^2$ penalizes hedge mismatch at maturity. Throughout the paper, we will assume $\rho > 0$ is a fixed constant.
	\item \textbf{Arbitrageurs}, with inventory $\tilde{Q} = \tilde{Q}(0) + \int \tilde{q}\,ds$ and trading rate $\tilde{q}$ in $\mathbb{H}^2$, who equate instantaneous gains and marginal costs, that is, $c(t)\tilde{q}(t) = \mu(t)$.
\end{itemize}

The equilibrium we consider is of Arrow--Debreu type. Since investor optimality is already encoded by the demand function $F \mapsto \mathfrak{d}(F)$, the equilibrium conditions only need to specify retailer optimality, arbitrageur optimality, and clearing in the forward and spot markets.
\begin{definition}
	An equilibrium is a tuple $(F^\ast,S^\ast,q^\ast,\tilde q^\ast,\mathfrak{s}^\ast)$ such that:
	\begin{enumerate}[label=\textbf{(\roman*)}]
		\item $S^\ast$ has the form
		      \begin{equation*}
			      S^\ast(t)=S^\ast(0)+\int_0^t \mu^\ast(s)\,ds+\int_0^t \sigma^\ast\,dX,
			      \qquad S^\ast(T)=\mathfrak{G},
		      \end{equation*}
		      for some $\mu^\ast \in \mathbb{S}^\infty$ and $\sigma^\ast \in \mathbb{H}_{BMO}$;
		\item retailer optimality holds:
		      \begin{equation*}
			      J(q^\ast,\mathfrak{s}^\ast) \ge J(q,\mathfrak{s}^\ast)
		      \end{equation*}
		      for every admissible spot strategy $q$;
		\item arbitrageur optimality holds:
		      \begin{equation*}
			      c\tilde q^\ast=\mu^\ast;
		      \end{equation*}
		\item the forward and spot markets clear:
		      \begin{equation*}
			      \mathfrak{s}^\ast=\mathfrak{d}(F^\ast),
			      \qquad
			      Q^\ast(t)+\tilde Q^\ast(t)=M(0)+\int_0^t m(s)\,ds,
		      \end{equation*}
		      where $Q^\ast=\int q^\ast\,ds$ and $\tilde Q^\ast=\tilde Q^\ast(0)+\int \tilde q^\ast\,ds$.
	\end{enumerate}
\end{definition}

We can now state the two main equilibrium results of the paper.
We start with the general small-risk-aversion result, which gives the local equilibrium structure around the benchmark case.
We then turn to the no-risk-aversion case $\phi=0$, which is the most explicit and economically transparent benchmark.

\begin{theorem}[Small-Risk-Aversion Equilibrium]\label{thm:small_risk_aversion}
	Assume that $\rho^2 T^2\|1/c\|_\infty^2 < 1$.
	Then there exists $\phi_0 > 0$ such that, for every $\phi \in [0,\phi_0]$, there exists a unique equilibrium.
	It is characterized by the coupled FBSDE system
	\begin{equation}\label{eq:equilibrium_system}
		\begin{cases}
			S(t)         & = \displaystyle \mathfrak{G} - \int_t^T \mu\,ds - \int_t^T \sigma\,dX                                            \\
			Q(t)         & = \displaystyle \int_0^t q\,ds                                                                                   \\
			c(t)q(t)     & = \displaystyle -\rho\left(Q(T)- \mathfrak{s}\right) + \int_t^T \left(\mu-\phi\sigma^2Q\right)ds -\int_t^T Z\,dX \\
			\mu(t)       & = \displaystyle c(t)(m(t) - q(t))                                                                                \\
			E[Q(T)]      & = \displaystyle \mathfrak{s} - \frac{1}{\rho}\left(F-E[\mathfrak{G}]\right)                                      \\
			\mathfrak{s} & = \displaystyle \mathfrak{d}(F).
		\end{cases}
	\end{equation}
	Moreover, with $cq = \rho(\Lambda-PQ)$, the equilibrium is equivalently characterized by
	\begin{equation*}
		\begin{cases}
			S(t)         & = \displaystyle \mathfrak{G} -\int_t^T \left(cm - \rho\left(\Lambda - PQ\right)\right)dt - \int_t^T \sigma dX \\
			E[Q(T)]      & = \displaystyle \mathfrak{s} - \frac{1}{\rho}\left(F - E[\mathfrak{G}]\right)                                 \\
			\mathfrak{s} & = \displaystyle \mathfrak{d}(F)
		\end{cases}
	\end{equation*}
	where $P$, $\Lambda$, and $Q$ satisfy the Riccati system
	\begin{equation*}
		\begin{cases}
			P(t)       & = \displaystyle 1 -\int_{t}^{T} \left(P\left(\frac{\rho P}{c} +1  \right) - \frac{\phi}{\rho}\sigma^2\right) ds - \int_t^T Z^P dX          \\
			\Lambda(t) & = \displaystyle \mathfrak{s} - \int_t^T \left( \Lambda \left(\frac{\rho P}{c} + 1\right) -\frac{cm}{\rho}\right)dt - \int_t^T Z^\Lambda dX \\
			Q(t)       & = \displaystyle \rho \int_0^t \left(\frac{\Lambda}{c} - \frac{P}{c} Q\right) ds.
		\end{cases}
	\end{equation*}
\end{theorem}

\begin{theorem}[No-Risk-Aversion Equilibrium]\label{thm:existence_equilibrium_no_risk_aversion}
	If $\phi = 0$, there exists a unique equilibrium regardless of $\rho$ or $T$ and characterized by the coupled FBSDE system \eqref{eq:equilibrium_system}.
	The Ricatti system is then decoupled from $\sigma$ and simplifies to
	\begin{equation*}
		\begin{cases}
			P(t)       & = \displaystyle 1 -\int_{t}^{T} P\left(\frac{\rho P}{c} +1  \right) ds - \int_t^T Z^P dX                                                   \\
			\Lambda(t) & = \displaystyle \mathfrak{s} - \int_t^T \left( \Lambda \left(\frac{\rho P}{c} + 1\right) -\frac{cm}{\rho}\right)dt - \int_t^T Z^\Lambda dX \\
			Q(t)       & = \displaystyle \rho \int_0^t \left(\frac{\Lambda}{c} - \frac{P}{c} Q\right) ds.
		\end{cases}
	\end{equation*}
	In particular,
	\begin{equation*}
		F-S(0) = \rho \Lambda(0).
	\end{equation*}
\end{theorem}

The no-risk-aversion case is sufficiently explicit to yield semi-closed formulas and to make the equilibrium mechanism fully transparent.
Once $P$ is known, the BSDE for $\Lambda$ is linear, hence affine in the forward supply $\mathfrak{s}$, while the equation for $Q$ is a linear stochastic ODE driven by $\Lambda$.
As a result, both $\Lambda$ and $Q$ are affine in $\mathfrak{s}$, which in turn implies that the equilibrium supply curve $\mathfrak{s}(F)$ is itself affine in the forward price $F$.
To make this structure explicit, introduce
\begin{align*}
	H(t)      & = \int_0^t \left(\frac{\rho P}{c} + 1\right)ds,
	          &
	\delta(t) & = e^{H(t)}E_t\left[\int_t^T e^{-H}cm\,ds\right],
	\\
	\alpha(t) & = e^{-H(t)+t}\int_0^t \frac{e^{-s}}{c} e^{H}P\,ds,
	          &
	\beta(t)  & = e^{-H(t)+t}\int_0^t \frac{e^{-s}}{c} e^{H} \delta\,ds.
\end{align*}
Then the linear BSDE for $\Lambda$ and the linear ODE for $Q$ can be solved explicitly as
\begin{equation*}
	\Lambda(t)=P(t)\mathfrak{s}+\frac{\delta(t)}{\rho},
	\qquad
	Q(t)=\rho\alpha(t)\mathfrak{s}+\beta(t).
\end{equation*}
The term $\delta$ captures the contribution of the supply rate $m$ to the retailer's continuation value, while $\beta$ is the corresponding contribution to the inventory process.
Substituting the affine form of $Q(T)$ into the clearing condition
\begin{equation*}
	E[Q(T)] = \mathfrak{s} - \frac{1}{\rho}(F-E[\mathfrak{G}])
\end{equation*}
gives
\begin{equation*}
	\rho E[\alpha(T)]\,\mathfrak{s} + E[\beta(T)] = \mathfrak{s} - \frac{1}{\rho}(F-E[\mathfrak{G}]).
\end{equation*}
Rearranging,
\begin{equation*}
	\bigl(1-\rho E[\alpha(T)]\bigr)\mathfrak{s} = E[\beta(T)] + \frac{1}{\rho}(F-E[\mathfrak{G}]).
\end{equation*}
To identify the slope explicitly, note that
\begin{equation*}
	\alpha(T)=e^{H(T)+T}\int_0^T \frac{e^{-s}}{c}e^{H}P\,ds
\end{equation*}
and, from the Riccati equation for $P$,
\begin{equation*}
	de^{-t}P = \rho e^{-t}\frac{P^2}{c}dt + e^{-t}Z^P\,dX.
\end{equation*}
Taking expectations and integrating from $0$ to $T$ gives
\begin{equation*}
	\rho E\left[\int_0^T \frac{e^{-s}}{c}P^2\,ds\right] = e^{-T} - P(0).
\end{equation*}
Since $P(t)=e^{H(t)}E_t[e^{-H(T)}]$, the same conditioning argument as in the proof below yields
\begin{equation*}
	1-\rho E[\alpha(T)] = 1-\rho e^T E\left[\int_0^T \frac{e^{-s}}{c}P^2\,ds\right] = e^T P(0)>0.
\end{equation*}
resulting into an affine strictly decreasing endogenous supply curve
\begin{equation*}
	\mathfrak{s}(F) = e^{-T}\frac{E[\beta(T)]}{P(0)} + \frac{e^{-T}}{\rho P(0)}\big(F-E[\mathfrak{G}]\big).
\end{equation*}
Thus higher forward prices induce greater dealer supply, and the equilibrium forward price is obtained by intersecting this increasing affine supply curve with the decreasing investor demand curve $\mathfrak{d}(F)$.
In particular, when demand is perfectly inelastic, market clearing fixes $\mathfrak{s}=\bar{\mathfrak{d}}$, so the affine supply relation can be solved immediately for the equilibrium forward price.
\begin{corollary}[Constant Forward Demand]\label{cor:linear_demand}
	Assume $\phi=0$ and $\mathfrak{d}(F) \equiv \bar{\mathfrak{d}}$.
	Then the equilibrium forward and spot prices at time $0$ are given by
	\begin{align*}
		F    & = \displaystyle E[\mathfrak{G}] + \rho \left[e^TP(0)\bar{\mathfrak{d}} - E[\beta(T)]\right],
		\\
		S(0) & = \displaystyle E[\mathfrak{G}] +(1-e^{-T})\rho \left[e^TP(0)\bar{\mathfrak{d}} - E[\beta(T)]\right]  - \rho e^{-T}E[\beta(T)]- \delta(0),
	\end{align*}
	so that
	\begin{equation*}
		F - S(0) = \rho P(0)\bar{\mathfrak{d}} + \delta(0).
	\end{equation*}
	In the homogeneous case $m \equiv 0$, one has $\delta \equiv 0$ and $\beta \equiv 0$, and the formulas reduce to
	\begin{align*}
		F    & = \displaystyle E[\mathfrak{G}] + \rho e^TP(0)\bar{\mathfrak{d}},
		\\
		S(0) & = \displaystyle E[\mathfrak{G}] +(1-e^{-T})\rho e^TP(0)\bar{\mathfrak{d}},
	\end{align*}
	with
	\begin{equation*}
		F - S(0) = \rho P(0)\bar{\mathfrak{d}}.
	\end{equation*}
\end{corollary}
\begin{remark}
	In the no-risk-aversion case, the forward-spot spread admits a natural economic decomposition into a hedging-cost component and a cost-weighted supply component.
	Indeed, as the Riccati equation shows, $P(0)$ depends only on the transaction-cost process $c$ and not on the supply rate $m$.
	Moreover, $P(0)$ is increasing in $c$: higher spot trading costs make it more expensive for a dealer to dynamically build and unwind the hedge associated with a forward position, so the term $\rho P(0)\mathfrak{s}$ raises the premium of the forward over the spot price.
	By contrast, once $P$ is known, the second component is given by
	\begin{equation*}
		\delta(0) = E\left[\int_0^T e^{-H(s)}c(s)m(s)\,ds\right],
	\end{equation*}
	so it depends on the supply rate only through the product $cm$.
	This shows that supply does not matter in isolation: what matters for pricing is the cost of absorbing that supply in the spot market.
	For fixed $c$, the quantity $\delta(0)$ is linear in $m$, so a larger supply rate increases the spread only insofar as dealers and arbitrageurs must intermediate larger inventories.
	Conversely, for fixed $m$, larger transaction costs amplify the contribution of supply to the spread.
	Thus higher transaction costs increase the spread directly through hedging frictions and indirectly by magnifying the price impact of supply.
	In this sense, the relevant state variable for forward pricing is not supply alone, but cost-weighted supply.
\end{remark}

\section{CNY/CNH Application}
We apply the no-risk-aversion representation to the CNY/CNH forward wedge.
The onshore market is denoted by $Y$ and the offshore market by $H$.
Both venues trade the same terminal RMB payoff $\mathfrak G$; forward demand is common and constant, $\mathfrak d^Y=\mathfrak d^H=\bar{\mathfrak d}$; and the spot supply rate is the same constant $m$ in both venues.
The application therefore attributes the wedge to heterogeneous hedging conditions, not to cross-market differences in forward demand or supply.

For each venue $i\in\{Y,H\}$, the no-risk-aversion representation gives the forward--spot premium
\begin{equation*}
	F^i-S^i(0)=\rho P^i(0)\bar{\mathfrak d}+\delta^i(0).
\end{equation*}
Here $P^i(0)$ is the marginal dynamic hedging coefficient and $\delta^i(0)$ is the cost-weighted inventory component.
Both are venue-specific because both depend on the local spot-trading cost process.
With $\Delta X:=X^Y-X^H$, spot parity is the constraint
\begin{equation*}
	(e^T-1)\rho\Delta P(0)\bar{\mathfrak d}-\rho E[\Delta\beta(T)]-\Delta\delta(0)=0.
\end{equation*}
It equates the spot effect of the difference in dynamic hedging costs, represented by $\Delta P(0)$ and $\Delta\beta(T)$, with the difference in cost-weighted inventory terms.
Under the same constraint, the forward differential is
\begin{equation*}
	F^Y-F^H=\rho\Delta P(0)\bar{\mathfrak d}+\Delta\delta(0).
\end{equation*}
Thus equality of spot prices restricts the hedging and inventory coefficients, but it leaves the forward wedge equal to the difference between two equilibrium forward premia.
The empirical relation $F^Y<F^H$ means that the offshore forward premium must be larger than the onshore one, even though offshore trading is cheaper in ordinary conditions.

The deterministic benchmark makes this tension explicit.
Assume constant costs with the normal ordering
\begin{equation*}
	c^H<c^Y.
\end{equation*}
The large-$\rho$ regime corresponds to a retailer facing a strong terminal mismatch penalty, so that the forward position is hedged as a contract to be delivered rather than left to be covered by a terminal spot purchase.
In this regime, the formulas derived in Section 4 imply that spot parity yields
\begin{equation*}
	F^Y-F^H \sim m(c^Y-c^H).
\end{equation*}
Since $m>0$ and $c^Y>c^H$, the deterministic model predicts $F^Y>F^H$.
The benchmark therefore fails as empirically the onshore forward is persistently cheaper than the offshore one.

Random offshore stress changes the relevant marginal hedging cost without breaking spot parity.
We keep the onshore cost constant and let the offshore cost jump at the first arrival time $\tau$ of a Poisson process with intensity $\lambda$:
\begin{equation*}
	c^Y(t)\equiv c^Y,
	\qquad
	c^H(t)=\underline{c}^H\mathbf{1}_{\{t<\tau\}}+\overline{c}^H\mathbf{1}_{\{t\geq\tau\}},
	\qquad
	\underline{c}^H<c^Y<\overline{c}^H.
\end{equation*}
The offshore venue is cheaper in the normal state, $\underline{c}^H<c^Y$, but becomes more expensive in the stress state, $\overline{c}^H>c^Y$.
At time $0$, before any stress event, the offshore hedging coefficient already incorporates the possibility of the future transition.
We write $\underline{P}^H$ for the offshore coefficient in the normal pre-stress state and $\overline{P}^H$ for the corresponding coefficient after the stress event.
Section 4 gives the following large-$\rho$ asymptotics for the scaled normal-state coefficient:
\begin{equation*}
	(\rho\underline{P}^H)'(t)
	\sim
	\rho\underline{P}^H(t)\left(\frac{\rho\underline{P}^H(t)}{\underline{c}^H}+1+\lambda\right)
	-\lambda\frac{\overline{c}^H e^{t-T}}{1-e^{t-T}},
	\qquad
	\rho\overline{P}^H(t)\sim\frac{\overline{c}^H e^{t-T}}{1-e^{t-T}}.
\end{equation*}
Similarly, writing $\underline{\delta}^H$ and $\overline{\delta}^H$ for the normal and stressed supply components,
\begin{equation*}
	\underline{\delta}^H(t)=e^{\underline{H}^H(t)}\int_t^T e^{-\underline{H}^H(s)}\left(\lambda\overline{\delta}^H(s)+\underline{c}^H m\right)ds.
\end{equation*}
Here $\underline{H}^H(t)\sim\int_0^t(\rho\underline{P}^H(s)/\underline{c}^H+1+\lambda)ds$.
The stress-arrival terms involving $\lambda$ are precisely what is absent from the deterministic benchmark.
Combining the normal-state offshore coefficients with the deterministic onshore coefficients and imposing spot parity gives
\begin{equation*}
	F^Y-F^H
	\sim
	\left(\frac{c^Y}{e^T-1}-\rho\underline{P}^H(0)\right)\bar{\mathfrak d}
	+\delta^Y(0)-\underline{\delta}^H(0).
\end{equation*}
When $\lambda=0$, the stress state is never reached and $\rho\underline{P}^H(0)\sim\underline{c}^H/(e^T-1)$, so this expression collapses to the deterministic empirically problematic formula above.
For $\lambda>0$, the offshore premium is instead governed by the whole state-contingent hedging technology.
The offshore forward can therefore be expensive today because dealers price the expected cost of hedging through a market that may become illiquid before maturity.

The calibration treats this mechanism as an inverse pricing problem.
For fixed $m$, $c^Y$, and normal offshore cost $\underline{c}^H$, each target level differential $x$ determines stress parameters through
\begin{equation*}
	S^Y(0)=S^H(0;\lambda,\overline{c}^H),
	\qquad
	\bigl(F^Y-F^H\bigr)(\lambda,\overline{c}^H)=x.
\end{equation*}
The intensity $\lambda$ corresponds to the market-implied stress probability $1-e^{-\lambda T}$ over the horizon, while $\overline{c}^H$ measures the market-implied severity of the offshore cost deterioration.
These are pricing-implied quantities, not physical estimates of crisis frequency.
The introduction reports annualized log forward ratios $r_T^{Y/H}$; the numerical exercise uses normalized one-year level targets $x=F^Y-F^H$ in the model units.
Figure \ref{fig:jump_target_calibration} and Table \ref{table:jump_target_calibration} report the resulting map $x\mapsto(\lambda,\overline{c}^H)$.
The deterministic benchmark corresponds to $\lambda=0$ and gives a positive forward differential; negative target wedges require positive stress intensity.
The calibration summarizes the mechanism: the empirical finding $F^Y<F^H$ is compatible with spot parity when the offshore forward embeds the risk of future CNH liquidity stress.

\begin{figure}[h!]
	\centering
	\includegraphics[width=\textwidth]{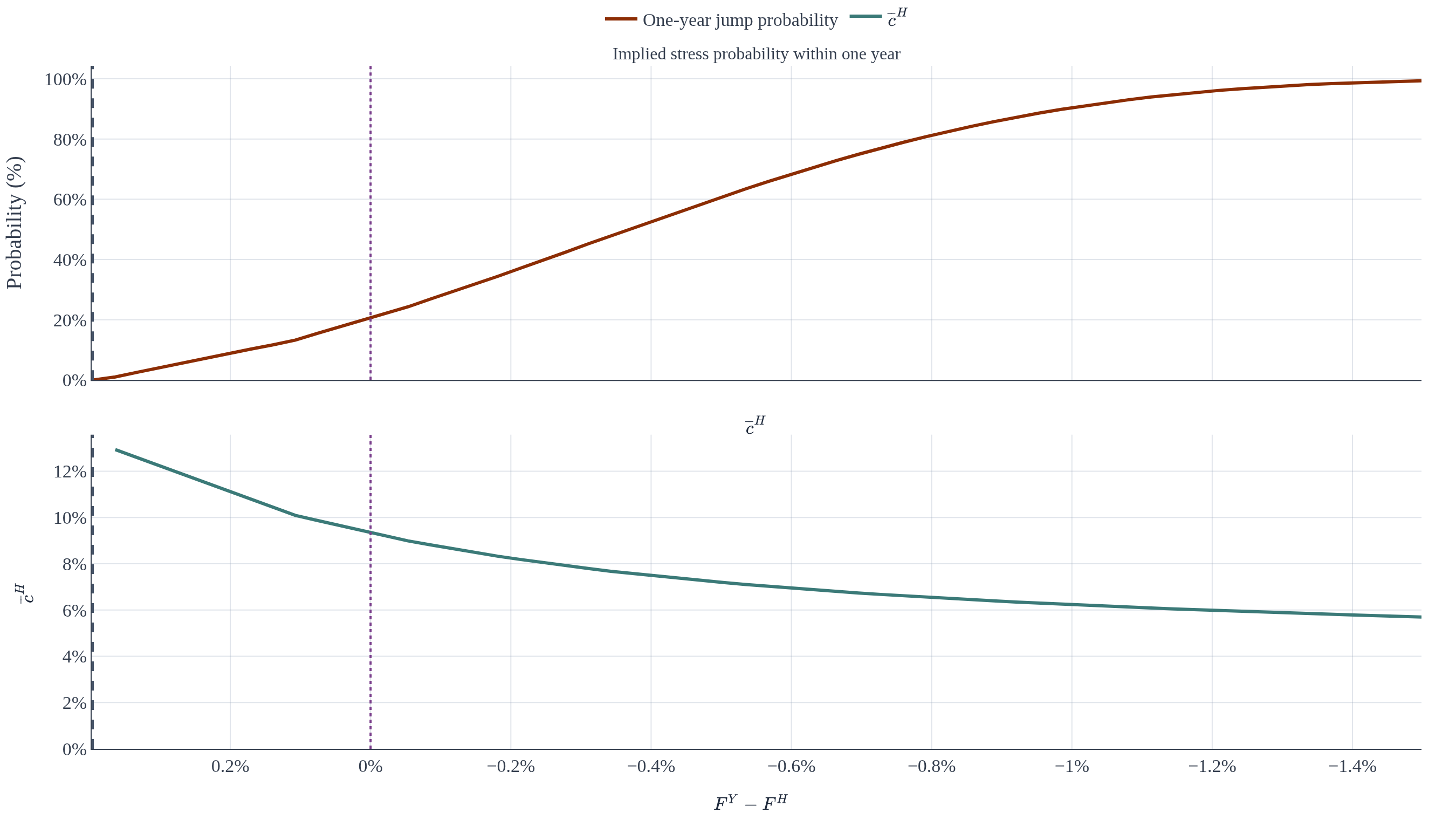}
	\caption{Inverse calibration of market-implied offshore stress. For each normalized one-year target level differential $F^Y-F^H$, the curve reports the pair $(\lambda,\overline{c}^H)$ that restores spot parity, holding fixed the common supply rate $m$, the onshore cost $c^Y$, and the normal-times offshore cost $\underline{c}^H$. The deterministic benchmark corresponds to $\lambda=0$.}
	\label{fig:jump_target_calibration}
\end{figure}

\begin{table}[h!]
	\begin{center}
		\begin{tabular}{@{}ccc@{}}
			\toprule
			Target $F^Y-F^H$ & Implied $\lambda$ & Implied $\overline{c}^H$  \\
			\midrule
			0.40\%           & 0.000             & irrelevant at $\lambda=0$ \\
			0.20\%           & 0.091             & 11.14\%                   \\
			0.11\%           & 0.142             & 10.09\%                   \\
			0.01\%           & 0.222             & 9.42\%                    \\
			-0.50\%          & 0.941             & 7.19\%                    \\
			-0.98\%          & 2.286             & 6.25\%                    \\
			\bottomrule
		\end{tabular}
		\caption{Implied jump beliefs consistent with spot parity and a normalized one-year target level differential. The deterministic benchmark corresponds to $\lambda=0$, in which case $\overline{c}^H$ is irrelevant.}
		\label{table:jump_target_calibration}
	\end{center}
\end{table}

\clearpage

\section{Proofs and Technical Derivations}

This final section is devoted to the proofs of the equilibrium results stated in Section 2 and to the technical derivations used in the CNY/CNH application.
We first isolate the retailer's optimization problem and the equilibrium reduction.
We then prove Theorem \ref{thm:existence_equilibrium_no_risk_aversion}, derive the specialized CNY/CNH systems used in Section 3, collect the technical lemmas used for the perturbative argument, and finally prove Theorem \ref{thm:small_risk_aversion}.

Assume for now that $c \in \mathbb{S}^\infty$ is strictly positive and bounded away from $0$, $\mu \in \mathbb{S}^\infty$, and $\sigma \in \mathbb{H}_{BMO}$.
\begin{proposition}[Retailer Optimization and Riccati Representation]
	\label{prop:retailer_fbsde}
	Let be given forward quantity $\mathfrak{s}$, forward price $F$ and spot price $S$ with $\mu$ in $\mathbb{S}^\infty$ and $\sigma$ in $\mathbb{H}_{BMO}$.
	The retailer's optimization problem admits a unique optimal strategy $q \in \mathbb{S}^\infty$ if and only if the FBSDE

	\begin{equation}\label{eq:retailer_fbsde}
		\begin{cases}
			Q(t)      & = \displaystyle \int_0^t q ds                                                                           \\
			c(t) q(t) & = \displaystyle -\rho(Q(T)-\mathfrak{s}) + \int_t^T \left( \mu -\phi \sigma^2Q \right)ds - \int_t^T ZdX
		\end{cases}
	\end{equation}
	has a unique solution with $q$ in $\mathbb{S}^\infty$ and $Z$ in $\mathbb{H}_{BMO}$.
	Moreover, this solution admits the Riccati representation
	\begin{equation*}
		cq = \rho(\Lambda - PQ)
	\end{equation*}
	where $(P, Z^P)$ and $(\Lambda, Z^\Lambda)$ are the unique solutions in $\mathbb{S}^\infty \times \mathbb{H}_{BMO}$ of the BSDEs
	\begin{equation}\label{eq:ricatti_retailer}
		\begin{cases}
			P(t)       & = \displaystyle 1 - \int_t^T \left(\frac{\rho P^2}{c} - \frac{\phi\sigma^2}{\rho}\right) ds - \int_t^T Z^P dX                  \\
			\Lambda(t) & = \displaystyle \mathfrak{s} - \int_t^T \left(\frac{\rho P}{c}\Lambda(s) - \frac{\mu}{\rho}\right) ds - \int_t^T Z^\Lambda dX.
		\end{cases}
	\end{equation}
\end{proposition}

\begin{proof}
	Since the objective function $\mathcal{J}$ is strictly convex in $q$, there exists a unique solution if and only if there exists $q$ setting the G\^ateaux derivative in any direction to $0$.
	Computing the G\^ateaux derivative of $\mathcal{J}$ in the direction $q^\prime$ yields
	\begin{multline*}
		\lim_{\varepsilon \searrow 0} \frac{\mathcal{J}(q + \varepsilon q^\prime, \mathfrak{s}) - \mathcal{J}(q, \mathfrak{s}) }{\varepsilon }
		\\
		= E \left[  -\rho \left( Q(T) - \mathfrak{s} \right) \left(\int_0^T  q^\prime dt\right) - \int_0^T c q q^\prime dt + \int_0^T  \left(\mu - \phi \tilde{\sigma}^2 Q  \right)  \left(\int_0^t q^\prime ds\right)dt\right]\\
		=  E  \left[\int_0^T E_t \left[  -\rho \left( Q(T) - \mathfrak{s} \right)  - cq + \int_t^T  \left(  \mu   - \phi \sigma^2  Q  \right)  ds  \right] q^\prime  dt \right]
	\end{multline*}
	where the last equality holds by Fubini's theorem and the tower property of the conditional expectation.
	Since this has to hold for any perturbation $q^\prime$, it is therefore equivalent to
	\begin{equation*}
		c(t)q(t) =  E_t \left[  -\rho \left( Q(T) - \mathfrak{s} \right) +   \int_t^T  \left(  \mu   -  \phi \sigma^2  Q\right)  ds  \right]
	\end{equation*}
	having a unique solution.
	For such a solution $q$, martingale representation yields
	\begin{equation*}
		c(t)q(t) = -\rho \left( Q(T) - \mathfrak{s} \right) +  \int_t^T  \left(  \mu - \phi \sigma^2  Q \right)  ds - \int_t^T  Z dX
	\end{equation*}
	for $Z$ in $\mathbb{H}^2$.
	By the integrability assumptions on the coefficients, it follows that $Z$ is in $\mathbb{H}_{BMO}$, yielding the desired FBSDE representation.

	Conversely, assume that the FBSDE has a unique solution with $q$ and $Z$ in $\mathbb{S}^\infty$ and $\mathbb{H}_{BMO}$, respectively; taking conditional expectations and rearranging yields the G\^ateaux derivative condition.

	We divide the proof of Riccati representation into four steps.

	\textit{Step~1: Riccati ansatz.}
	We seek a solution of the form $cq = \rho(\Lambda - PQ)$ with terminal values $P(T) = 1$ and $\Lambda(T) = \mathfrak{s}$.
	Applying It\^o's formula to $Y = cq$ and using the forward equation $dQ = q dt = \frac{\rho}{c}(\Lambda - P Q)dt$ yields
	\begin{equation*}
		dY = \rho d\Lambda - \rho Q dP - \frac{\rho^2 P}{c}(\Lambda - PQ)dt.
	\end{equation*}
	On the other hand, differentiating the backward equation in \eqref{eq:retailer_fbsde} gives
	\begin{equation*}
		dY = -(\mu - \phi\sigma^2 Q)dt + ZdX.
	\end{equation*}
	Identifying the drifts and martingale parts yields the system \eqref{eq:ricatti_retailer}.

	\textit{Step~2: Existence for $P$.}
	The BSDE for $P$ has terminal condition $P(T) = 1$ and driver
	\begin{equation*}
		f^P(P, Z^P) = \frac{\rho P^2}{c} - \frac{\phi\sigma^2}{\rho}.
	\end{equation*}
	Since $c$ is uniformly bounded away from zero, $f^P$ has quadratic growth in $P$.
	Moreover, $\sigma \in \mathbb{H}_{BMO}$ provides the required BMO integrability control for the quadratic term.
	By the same truncation and comparison argument used in Lemma \ref{lemma:riccati_solution}, there exists a unique solution $(P, Z^P)$ in $\mathbb{S}^\infty \times \mathbb{H}_{BMO}$.

	\textit{Step~3: Existence for $\Lambda$ and construction of $Q$.}
	Given $P$ from Step~2, the BSDE in $\Lambda$ is linear with bounded coefficients since $P$ is in $\mathbb{S}^\infty$ and $c$ is bounded away from zero.
	By standard linear BSDE theory with bounded coefficients, there exists a unique solution $(\Lambda, Z^\Lambda)$ in $\mathbb{S}^\infty \times \mathbb{H}_{BMO}$.

	With $P$ and $\Lambda$ determined, define $Q$ by the linear stochastic ODE
	\begin{equation*}
		Q(t) = \rho\int_0^t \frac{\Lambda - PQ}{c}ds.
	\end{equation*}
	This has the explicit solution
	\begin{equation*}
		Q(t) = \rho\int_0^t \exp\left(-\rho\int_s^t \frac{P}{c}du\right)\frac{\Lambda}{c}ds,
	\end{equation*}
	which belongs to $\mathbb{S}^\infty$ since $P$ and $\Lambda$ are in $\mathbb{S}^\infty$ and $c$ is bounded away from zero.

	\textit{Step~4: Verification and uniqueness.}
	Set $q = \frac{\rho(\Lambda - PQ)}{c}$ and $Z = \rho(Z^\Lambda - QZ^P)$.
	A direct computation using It\^o's formula shows that $(Q, q, Z)$ satisfies \eqref{eq:retailer_fbsde}.
	The terminal condition follows from $P(T) = 1$ and $\Lambda(T) = \mathfrak{s}$.

	For uniqueness, suppose $(Q, q, Z)$ and $(\tilde{Q}, \tilde{q}, \tilde{Z})$ are two solutions.
	By the converse direction of Proposition~\ref{prop:retailer_fbsde}, both $q$ and $\tilde{q}$ are optimal strategies for the retailer's problem.
	Since the objective $\mathcal{J}$ is strictly convex, the optimal strategy is unique, hence $q = \tilde{q}$.
	This implies $Q = \tilde{Q}$ and, by the backward equation, $Z = \tilde{Z}$.
\end{proof}

\begin{remark}
	Although the retailer's FBSDE is now characterized for a given spot price, the full equilibrium system is substantially more difficult.
	Indeed, the coefficients $\mu$ and $\sigma$ are now coupled back into the spot-price BSDE, while the retailer equation itself remains quadratic through the term $\phi\sigma^2 Q$.

	Even if the dynamic problem were solved for a prescribed forward supply $\mathfrak{s} = \mathfrak{d}(F)$, one would still need to solve the static clearing condition.
	This introduces the additional constraint $E[Q(T)] = \mathfrak{s} - (F-E[\mathfrak{G}])/\rho$, so the forward price must be determined through the law of the dynamic solution.

	The strategy is therefore to first solve the case $\phi = 0$, where the quadratic coupling disappears, and then use the resulting estimates in a perturbative argument for small $\phi$.
\end{remark}
\subsection{Proof of Theorem \ref{thm:existence_equilibrium_no_risk_aversion}}
Throughout this subsection, we assume that $\phi = 0$ and prove the no-risk-aversion equilibrium theorem.

\begin{proof}
	In the case where $\phi = 0$, imposing the arbitrageur condition and market clearing on the retailer problem yields the coupled equilibrium system
	\begin{equation*}
		\begin{cases}
			S(t)         & = \displaystyle \mathfrak{G}-\int_t^T \mu ds - \int_t^T \sigma dX                     \\
			Q(t)         & = \displaystyle \int_0^t q ds                                                         \\
			c(t)q(t)     & = \displaystyle -\rho\left(Q(T)- \mathfrak{s}\right) + \int_t^T \mu ds -\int_t^T Z dX \\
			\mu          & = \displaystyle c(m - q)                                                              \\
			E[Q(T)]      & = \mathfrak{s} - \frac{1}{\rho}\left(F-E[\mathfrak{G}]\right)                         \\
			\mathfrak{s} & = \mathfrak{d}(F),
		\end{cases}
	\end{equation*}
	where the arbitrageur condition gives $\tilde q = m-q$.

	We make the separation of variables ansatz $cq = \rho(\Lambda - PQ)$ where $P$ and $\Lambda$ are It\^o processes with terminal values $P(T) = 1$ and $\Lambda(T) = \mathfrak{s}$.
	From It\^o's lemma, and since $dQ = q\,dt = \frac{\rho}{c}(\Lambda - PQ)\,dt$ is of bounded variation, it holds
	\begin{align*}
		d(\Lambda - PQ) & = d\Lambda - Q\,dP - P\,dQ                                                                                                         \\
		                & = \displaystyle \left(E^\Lambda - Q E^P - \rho\frac{P}{c}\Lambda + \rho\frac{P^2}{c}Q\right)dt + \left(Z^\Lambda - Z^{P}Q\right)dX
	\end{align*}
	On the other hand, differentiating the backward equation and using $\mu = c(m-q)$ yields
	\begin{align*}
		d(\Lambda - PQ) & = -\frac{\mu}{\rho}\,dt + \frac{Z}{\rho}\,dX                                      \\
		                & = \left(\left(\Lambda - PQ\right) - \frac{c}{\rho}m\right)dt + \frac{Z}{\rho}\,dX
	\end{align*}
	Identifying the drifts and martingale parts while collecting terms by powers of $Q$ gives
	\begin{equation*}
		\begin{cases}
			\displaystyle E^P            & = \displaystyle P\left(\frac{\rho P}{c} + 1\right)                        \\
			\displaystyle E^\Lambda      & = \displaystyle \Lambda\left(\frac{\rho P}{c} +1\right) - \frac{c}{\rho}m \\
			\displaystyle \frac{Z}{\rho} & = Z^\Lambda - Z^{P}Q
		\end{cases}
	\end{equation*}
	Showing that $P$, $\Lambda$ and $Q$ satisfy
	\begin{equation*}
		\begin{cases}
			P(t)       & = \displaystyle 1 -\int_t^T P\left(\frac{\rho P}{c} + 1\right) ds-\int_t^T Z^P\,dX                                                      \\
			\Lambda(t) & = \displaystyle \mathfrak{s} -\int_t^T\left(\Lambda\left(\frac{\rho P}{c} +1\right) - \frac{cm}{\rho}\right)ds - \int_t^T Z^\Lambda\,dX \\
			Q(t)       & = \displaystyle \rho \int_0^t \left(\frac{\Lambda}{c} - \frac{P}{c}Q \right) ds
		\end{cases}
	\end{equation*}

	Since $c$ is bounded away from zero and $\rho > 0$, the driver of the BSDE for $P$ has quadratic growth in $P$.
	By the $\phi=0$ case of Lemma \ref{lemma:riccati_solution}, there exists a unique solution $P \in \mathbb{S}^\infty$ with $0 < P < 1$ on $[0,T)$.
	Given $P$, the BSDE for $\Lambda$ is linear with bounded coefficients, hence admits a unique solution $\Lambda \in \mathbb{S}^\infty$.
	Finally, $Q$ is obtained from the linear ODE above, which has the explicit solution
	\begin{equation*}
		Q(t) = \rho\int_0^t \exp\left(-\rho\int_s^t \frac{P}{c}\,du\right)\frac{\Lambda}{c}\,ds.
	\end{equation*}

	The quantity
	\begin{equation*}
		H(t) = \int_0^t \left(\frac{\rho P}{c} + 1\right) ds
	\end{equation*}
	plays a central role.
	Denoting $Y(t) = E_t[e^{-(H(T)-H(t))}]$ it follows that $dY = Y(\rho P/c + 1)\,dt + Z\,dX$ with $Y(T) = 1$, showing that $Y = P$ and therefore
	\begin{equation*}
		P(t) = e^{H(t)} E_t\left[e^{-H(T)}\right].
	\end{equation*}
	Furthermore, by It\^o's lemma it follows that $d(e^{-t}P) = \rho e^{-t}P^2/c\,dt + e^{-t}Z\,dX$ showing that
	\begin{equation*}
		\rho E_t\left[\int_t^T \frac{e^{-s}}{c} P^2 ds\right] = e^{-T} - e^{-t}P(t).
	\end{equation*}

	Since $P$ is uniformly bounded and $\Lambda$ is an affine BSDE, it has the explicit solution
	\begin{equation*}
		\Lambda(t) = e^{H(t)}E_t\left[ \mathfrak{s}e^{-H(T)} + \int_t^T e^{-H}\frac{cm}{\rho}ds \right]
		= P(t)\mathfrak{s} + \frac{\delta(t)}{\rho}
	\end{equation*}
	where
	\begin{equation*}
		\delta(t) = e^{H(t)}E_t\left[\int_t^T e^{-H} cm\,ds\right].
	\end{equation*}

	With $Q(0) = 0$, given that $P$ and $\Lambda$ are uniformly bounded, it follows that
	\begin{equation*}
		Q(t) = \rho e^{-H(t)+t}\int_0^t e^{H - s}\frac{\Lambda}{c}\,ds = \rho \alpha(t) \mathfrak{s} + \beta(t)
	\end{equation*}
	where
	\begin{align*}
		\alpha(t) & = e^{-H(t)+t}\int_0^t \frac{e^{-s}}{c}e^{H}P\,ds
		          &
		\beta(t)  & = e^{-H(t)+t}\int_0^t \frac{e^{-s}}{c}e^{H}\delta\,ds.
	\end{align*}

	In particular, given that $P$ only depends on $\rho$ and $c$, the function $\Lambda$ depends on $c$, $\rho$ and $m$ and is affine in $\mathfrak{s}$, and so is $Q$.
	Since
	\begin{equation*}
		E[Q(T)] = \mathfrak{s} - \frac{F-E[\mathfrak{G}]}{\rho} = \rho E[\alpha(T)] \mathfrak{s} + E[\beta(T)]
	\end{equation*}
	it follows that
	\begin{equation*}
		\mathfrak{s}\left(1- \rho E[\alpha(T)]\right) = E[\beta(T)] + \frac{F-E[\mathfrak{G}]}{\rho}.
	\end{equation*}

	However
	\begin{multline*}
		1-\rho E[\alpha(T)]
		= 1-\rho e^{T}E\left[\int_0^T \frac{e^{-s}}{c}E_s\left[e^{-(H(T) - H)}\right] P ds\right]
		\\
		= 1- \rho e^{T}E\left[\int_0^T \frac{e^{-s}}{c}P^2 ds\right]
		= e^{T}P(0)> 0
	\end{multline*}
	Hence, we deduce that
	\begin{equation*}
		\mathfrak{s}(F) = e^{-T}\frac{E[\beta(T)]}{P(0)} + \frac{e^{-T}}{\rho P(0)}\left(F - E[\mathfrak{G}]\right)
	\end{equation*}
	which is a strictly increasing linear function of $F$.
	According to the assumptions on $\mathfrak{d}(F)$ which is strictly decreasing, there exists a unique $F$ satisfying $\mathfrak{d}(F) = \mathfrak{s}(F)$.

	As for the spot and forward together we have
	\begin{align*}
		S(t) & = \mathfrak{G} -\int_t^T \mu ds -\int_t^T \sigma dX                                                                     \\
		     & = \mathfrak{G} +\rho \int_t^T d(\Lambda - PQ) -\int_t^T \left(\sigma - Z\right) dX                                      \\
		     & = \mathfrak{G} +\rho \left(\mathfrak{s} - Q(T)\right) - \rho\left(\Lambda(t) - P(t)Q(t)\right) - \int_t^T (\sigma-Z) dX
	\end{align*}
	where we used $P(T) = 1$ and $\Lambda(T) = \mathfrak{s}$.
	In particular,
	\begin{equation*}
		S(0) = E[\mathfrak{G}] + \rho\left(\mathfrak{s} - E[Q(T)]\right) - \rho \Lambda(0) = F - \rho\Lambda(0)
	\end{equation*}
	showing that
	\begin{equation*}
		F-S(0) = \rho \Lambda(0)
		= \rho P(0)\mathfrak{s} + \delta(0)
		= \rho P(0)\mathfrak{s} + E\left[\int_0^T e^{-H}cm ds\right].
	\end{equation*}
\end{proof}
This completes the proof of the benchmark theorem.
The semi-explicit formulas for $H$, $\delta$, $\alpha$, $\beta$, as well as the affine supply curve $\mathfrak{s}(F)$, were stated already in Section 2 and are exactly the objects identified in the argument above.

\subsection{Derivation of the CNY/CNH Special Case}
We collect the deterministic formulas and the regime-switching reduction used in Section 3.

\textit{Deterministic benchmark.}
With constant transaction cost $c$ and constant supply rate $m$, the no-risk Riccati system is deterministic and reduces to
\begin{equation*}
	P^\prime(t) = P(t)\left(\frac{\rho P(t)}{c} + 1\right), \qquad P(T)=1,
\end{equation*}
whose solution is
\begin{equation*}
	P(t) = \frac{ce^{t-T}}{c + (1-e^{t-T})\rho}.
\end{equation*}
Using $\Lambda = P\bar{\mathfrak{d}} + \delta/\rho$ and $Q = \rho\alpha\bar{\mathfrak{d}} + \beta$, direct integration yields
\begin{align*}
	\delta(t) & = cm \frac{(c+\rho)(1- e^{t-T}) - \rho e^{t-T}(T-t)}{c + (1-e^{t-T})\rho},   \\
	\beta(t)  & = m\left[t-\frac{(c+\rho T)\left(e^{t}-1\right)}{(c+\rho)e^{T}-\rho}\right].
\end{align*}
Therefore, as $\rho\to\infty$,
\begin{align*}
	\rho e^TP(0)  & = \frac{c}{(1-e^{-T}) + c/\rho } \sim \frac{c}{1-e^{-T}},                                                                     \\
	\rho \beta(T) & = \frac{\rho cm}{c + (1-e^{-T})\rho} \left[T - (1-e^{-T})\right] \sim \frac{cm}{1-e^{-T}}\left[T - (1-e^{-T})\right],         \\
	\delta(0)     & = cm\frac{(c+\rho)(1-e^{-T}) - \rho e^{-T}T}{c + (1-e^{-T})\rho } \sim \frac{cm}{1- e^{-T}} \left[1 -e^{-T} - Te^{-T}\right].
\end{align*}
These are the asymptotic formulas used in Section 3.

\textit{Jump specialization.}
Set $I_t := \mathbf{1}_{\{\tau \leq t\}}$ and write
\begin{equation*}
	P(t) = (1-I_t)\underline{P}^H(t) + I_t\overline{P}^H(t),
	\qquad
	\Lambda(t) = (1-I_t)\underline{\Lambda}^H(t) + I_t\overline{\Lambda}^H(t),
\end{equation*}
with $c(t) = \underline{c}^H(1-I_t) + \overline{c}^H I_t$.
Since $dI_t = \lambda(1-I_t)dt + dM_t$ for a martingale $M$, applying It\^o's formula to the no-risk equations for $P$ and $\Lambda$ yields the stress-state equations
\begin{equation*}
	(\overline{P}^H)'(t)=\overline{P}^H(t)\left(\frac{\rho\overline{P}^H(t)}{\overline{c}^H}+1\right),
	\qquad
	(\overline{\Lambda}^H)'(t)=\overline{\Lambda}^H(t)\left(\frac{\rho\overline{P}^H(t)}{\overline{c}^H}+1\right)-\frac{\overline{c}^H m}{\rho},
\end{equation*}
and the normal-state equations
\begin{equation*}
	(\underline{P}^H)'(t)+\lambda\bigl(\overline{P}^H(t)-\underline{P}^H(t)\bigr)=\underline{P}^H(t)\left(\frac{\rho\underline{P}^H(t)}{\underline{c}^H}+1\right),
\end{equation*}
\begin{equation*}
	(\underline{\Lambda}^H)'(t)+\lambda\bigl(\overline{\Lambda}^H(t)-\underline{\Lambda}^H(t)\bigr)=\underline{\Lambda}^H(t)\left(\frac{\rho\underline{P}^H(t)}{\underline{c}^H}+1\right)-\frac{\underline{c}^H m}{\rho}.
\end{equation*}
Multiplying by $\rho$ gives the equation for the scaled coefficient $\rho\underline{P}^H$; together with $\rho\overline{P}^H(t)\sim \overline{c}^H e^{t-T}/(1-e^{t-T})$, this yields the normal-state asymptotic relation used in Section 3.

Now write $\underline{\Lambda}^H=\underline{P}^H\bar{\mathfrak{d}}+\underline{\delta}^H/\rho$ and $\overline{\Lambda}^H=\overline{P}^H\bar{\mathfrak{d}}+\overline{\delta}^H/\rho$. Substituting this decomposition into the equations for $\underline{\Lambda}^H$ and $\overline{\Lambda}^H$ and using the equations for $\underline{P}^H$ and $\overline{P}^H$ gives
\begin{equation*}
	(\overline{\delta}^H)'(t)=\overline{\delta}^H(t)\left(\frac{\rho\overline{P}^H(t)}{\overline{c}^H}+1\right)-\overline{c}^H m,
\end{equation*}
and
\begin{equation*}
	(\underline{\delta}^H)'(t)=\underline{\delta}^H(t)\left(\frac{\rho\underline{P}^H(t)}{\underline{c}^H}+1+\lambda\right)-\lambda\overline{\delta}^H(t)-\underline{c}^H m.
\end{equation*}
Solving the linear normal-state equation by variation of constants yields the expression for $\underline{\delta}^H$ used in Section 3.

\begin{remark}\label{rem:P_tilde_equals_P}
	At time $0$ the offshore market is in the normal state. Since Corollary \ref{cor:linear_demand} gives $F-S(0)=\rho\Lambda(0)$, the decomposition
	\begin{equation*}
		\underline{\Lambda}^H(0)=\underline{P}^H(0)\bar{\mathfrak{d}}+\frac{\underline{\delta}^H(0)}{\rho}
	\end{equation*}
	implies
	\begin{equation*}
		F^H-S^H(0)=\rho\underline{P}^H(0)\bar{\mathfrak{d}}+\underline{\delta}^H(0).
	\end{equation*}
\end{remark}

\subsection{Technical Lemmas for Small Risk Aversion}
Throughout this subsection we assume $c \in \mathbb{S}^\infty$ is strictly positive and bounded away from zero, and $\sigma \in \mathbb{H}_{BMO}$.

\begin{lemma}[Existence, Uniqueness and Bounds for the Riccati System]\label{lemma:riccati_solution}
	Fix $\tilde\phi \geq 0$ and $R_\sigma < \infty$.
	For every $\phi \in [0,\tilde\phi]$ and every $\sigma \in \mathbb{H}_{BMO}$ with $\|\sigma\|_{\mathbb{H}_{BMO}} \leq R_\sigma$, the system
	\begin{equation}\label{eq:riccati_system}
		\begin{cases}
			P(t)       & = \displaystyle 1 - \int_t^T \left(P\left(\frac{\rho P}{c} + 1\right) - \frac{\phi}{\rho}\sigma^2\right) ds - \int_t^T Z^P dX              \\
			\Lambda(t) & = \displaystyle \mathfrak{s} - \int_t^T \left(\Lambda\left(\frac{\rho P}{c} + 1\right) - \frac{cm}{\rho}\right)ds - \int_t^T Z^\Lambda\,dX \\
			Q(t)       & = \displaystyle \rho\int_0^t \left(\frac{\Lambda}{c} - \frac{P}{c}Q\right)ds, \quad Q(0) = 0
		\end{cases}
	\end{equation}
	has a unique solution $(P, Z^P, \Lambda, Z^\Lambda, Q) \in (\mathbb{S}^\infty \times \mathbb{H}_{BMO})^2 \times \mathbb{S}^\infty$.
	Moreover, if $\bar P$ denotes the corresponding $\phi=0$ solution, then
	\begin{equation*}
		0 < \bar P(t) \leq P(t) \leq C_P(\tilde\phi,R_\sigma) := 1 + \frac{\tilde\phi}{\rho}R_\sigma^2,
	\end{equation*}
	and
	\begin{equation*}
		\|\Lambda\|_{\mathbb{S}^\infty} \leq C_\Lambda := 2|\mathfrak{s}| + 4\frac{\|cm\|_\infty T}{\rho}, \quad
		\|Q\|_{\mathbb{S}^\infty} \leq h_Q := \left\|\frac{\rho}{c}\right\|_\infty C_\Lambda T.
	\end{equation*}
\end{lemma}

\begin{proof}
	\textit{Existence and bounds for $P$.}
	For $K > 0$, consider the Lipschitz-truncated driver
	\begin{equation*}
		f_K(P) := \left(P^+ \wedge K\right)\left(\frac{\rho(P^+ \wedge K)}{c} + 1\right) - \frac{\phi}{\rho}\sigma^2.
	\end{equation*}
	Since $f_K$ is uniformly Lipschitz in $P$ and $\sigma \in \mathbb{H}_{BMO}$, the BSDE
	\begin{equation*}
		P^K(t) = 1 - \int_t^T f_K(P^K)\,ds - \int_t^T Z^{P,K}\,dX
	\end{equation*}
	has a unique solution $(P^K, Z^{P,K}) \in \mathbb{S}^\infty \times \mathbb{H}_{BMO}$ by standard Lipschitz BSDE theory.

	\textit{Lower bound and comparison direction.}
	Let $\bar{P}$ be the $\phi = 0$ solution. By the exponential transform and the same truncation method, $0 < \bar{P} \leq 1$.
	The BSDE is written in the convention $Y_t=\xi-\int_t^T f_s\,ds-\int_t^T Z_s\,dX_s$.
	Thus a smaller driver produces a larger solution.
	Since, along $\bar P$ and for $K \geq 1$,
	\begin{equation*}
		f_K(\bar P) = \bar P\left(\frac{\rho\bar P}{c}+1\right)-\frac{\phi}{\rho}\sigma^2
		\leq \bar P\left(\frac{\rho\bar P}{c}+1\right),
	\end{equation*}
	the comparison theorem gives $P^K \geq \bar P>0$.

	\textit{Upper bound.}
	Since $P^K$ is positive, the truncated positive part is non-negative and conditional expectation in the BSDE gives
	\begin{equation*}
		P^K(t)
		\leq \mathbb{E}_t\left[1+\frac{\phi}{\rho}\int_t^T\sigma^2ds\right]
		\leq 1+\frac{\tilde\phi}{\rho}R_\sigma^2
		=: C_P(\tilde\phi,R_\sigma).
	\end{equation*}

	\textit{Passage to the limit.}
	For $K > C_P(\tilde\phi,R_\sigma)$, since $P^K \in (0, C_P(\tilde\phi,R_\sigma)]$, the truncation $(P^K)^+ \wedge K = P^K$ is inactive.
	Therefore $(P^K, Z^{P,K})$ solves the original (untruncated) BSDE.
	Setting $(P, Z^P) := (P^K, Z^{P,K})$ for any $K > C_P(\tilde\phi,R_\sigma)$ yields the solution with $0 < \bar P \leq P \leq C_P(\tilde\phi,R_\sigma)$.

	\textit{Uniqueness.}
	If $(P_1, Z_1)$ and $(P_2, Z_2)$ are two bounded solutions, comparison with $\bar P$ and the conditional-expectation estimate above imply $0 < \bar P \leq P_i \leq C_P(\tilde\phi,R_\sigma)$.
	Set $\delta P = P_1 - P_2$ and $\delta Z = Z_1 - Z_2$.
	Then $\delta P$ satisfies
	\begin{equation*}
		\delta P(t) = -\int_t^T \delta P\left(\frac{\rho(P_1 + P_2)}{c} + 1\right)ds - \int_t^T \delta Z\,dX, \quad \delta P(T) = 0.
	\end{equation*}
	This is a linear BSDE with bounded coefficient $\alpha(s) = \rho(P_1+P_2)/c + 1$.
	Multiplying by the integrating factor $e^{\int_t^T \alpha(s)\,ds}$ gives
	\begin{equation*}
		\delta P(t)\,e^{\int_t^T \alpha\,ds} = -\int_t^T e^{\int_s^T \alpha\,du}\,\delta Z\,dX.
	\end{equation*}
	Taking conditional expectations and using the martingale property, $\delta P(t) = 0$ for all $t$, hence $P_1 = P_2$ and $Z_1 = Z_2$.

	\textit{Existence and bounds for $\Lambda$.}
	Since $P \in \mathbb{S}^\infty$ and $c$ is bounded away from zero, the coefficient $\rho P/c + 1$ is bounded.
	By standard linear BSDE theory, there exists a unique solution $(\Lambda, Z^\Lambda) \in \mathbb{S}^\infty \times \mathbb{H}_{BMO}$.
	Applying It\^o's formula to $|\Lambda(\tau)|^2$ for any stopping time $\tau \leq T$,
	\begin{equation*}
		|\Lambda(\tau)|^2 + \mathbb{E}_\tau\left[\int_\tau^T |Z^\Lambda|^2\,ds\right]
		= |\mathfrak{s}|^2 + 2\mathbb{E}_\tau\left[\int_\tau^T \Lambda\left(-\Lambda\left(\frac{\rho P}{c} + 1\right) + \frac{cm}{\rho}\right)ds\right].
	\end{equation*}
	Dropping the non-negative term $\mathbb{E}_\tau[\int_\tau^T |Z^\Lambda|^2\,ds]$ and using $-\Lambda^2(\rho P/c+1) \leq 0$,
	\begin{equation*}
		|\Lambda(\tau)|^2 \leq |\mathfrak{s}|^2 + 2\frac{\|cm\|_\infty}{\rho}\,\|\Lambda\|_{\mathbb{S}^\infty}\,T.
	\end{equation*}
	Applying $2ab \leq a^2/2 + 2b^2$ with $a = \|\Lambda\|_{\mathbb{S}^\infty}$ and $b = \|cm\|_\infty T/\rho$,
	\begin{equation*}
		\|\Lambda\|_{\mathbb{S}^\infty}^2 \leq |\mathfrak{s}|^2 + \frac{1}{2}\|\Lambda\|_{\mathbb{S}^\infty}^2 + 2\frac{\|cm\|_\infty^2 T^2}{\rho^2}.
	\end{equation*}
	Rearranging gives $\|\Lambda\|_{\mathbb{S}^\infty}^2 \leq 2|\mathfrak{s}|^2 + 4\|cm\|_\infty^2 T^2/\rho^2 \leq C_\Lambda^2$.

	\textit{Bounds for $Q$.}
	The ODE for $Q$ has the explicit solution
	\begin{equation*}
		Q(t) = \rho\,e^{-O(t)}\int_0^t e^{O(s)}\frac{\Lambda(s)}{c(s)}\,ds,
		\qquad O(t) = \rho\int_0^t \frac{P(u)}{c(u)}\,du.
	\end{equation*}
	Since $P > 0$, the function $O$ is increasing, and therefore for $0 \leq s \leq t$ we have $e^{-O(t)}e^{O(s)} \leq 1$.
	Hence
	\begin{equation*}
		|Q(t)| \leq \rho\int_0^t \frac{|\Lambda(s)|}{c(s)}\,ds \leq \left\|\frac{\rho}{c}\right\|_\infty \|\Lambda\|_{\mathbb{S}^\infty} T,
	\end{equation*}
	which gives $\|Q\|_{\mathbb{S}^\infty} \leq h_Q$.
\end{proof}

\begin{lemma}[Stability of the Riccati System]\label{lemma:riccati_stability}
	Fix $\tilde\phi > 0$ and $R_\sigma < \infty$.
	For $\phi \in [0,\tilde\phi]$ and $\sigma_1, \sigma_2 \in \mathbb{H}_{BMO}$ with $\|\sigma_i\|_{\mathbb{H}_{BMO}} \leq R_\sigma$, denote by $(P_i, \Lambda_i, Q_i)$ the corresponding solutions to \eqref{eq:riccati_system}.
	Then
	\begin{equation*}
		\|\delta P\|_{\mathbb{S}^\infty}^2
		\leq g^2\,\|\sigma_1 - \sigma_2\|_{\mathbb{H}_{BMO}}^2
	\end{equation*}
	where $g = 2^{3/2}(\phi/\rho)(\|\sigma_1\|_{\mathbb{H}_{BMO}}^2 + \|\sigma_2\|_{\mathbb{H}_{BMO}}^2)^{1/2}$.
	Furthermore,
	\begin{equation*}
		\|\delta\Lambda\|_{\mathbb{S}^\infty}^2
		\leq \tilde{C}\,g^2\,\|\sigma_1 - \sigma_2\|_{\mathbb{H}_{BMO}}^2
	\end{equation*}
	with $\tilde{C} = 4C_\Lambda^2\|\rho/c\|_\infty^2\,T^2$, and
	\begin{equation*}
		\|\delta Q\|_{\mathbb{S}^\infty}^2
		\leq \bar{C}\,g^2\,\|\sigma_1 - \sigma_2\|_{\mathbb{H}_{BMO}}^2
	\end{equation*}
	with $\bar{C} = 2\|\rho/c\|_\infty^2 T^2\left(\tilde{C}+\|\rho/c\|_\infty^2T^2 C_\Lambda^2\right)$.
	Moreover, on this BMO ball, $g \leq 4(\phi/\rho)R_\sigma$, and therefore $g \to 0$ as $\phi \to 0$.
\end{lemma}

\begin{proof}
	\textit{Stability of $P$.}
	Set $\delta P = P_1 - P_2$ and $\delta Z = Z^{P,1} - Z^{P,2}$.
	Applying It\^o's formula to $|\delta P(t)|^2$ for any stopping time $\tau \leq T$,
	\begin{align*}
		 & |\delta P(\tau)|^2 + \mathbb{E}_\tau\left[\int_\tau^T |\delta Z|^2\,ds\right]                                                                                      \\
		 & = \mathbb{E}_\tau\left[2\int_\tau^T \delta P\left(-\delta P\left(\frac{\rho(P_1 + P_2)}{c} + 1\right) + \frac{\phi}{\rho}(\sigma_1^2 - \sigma_2^2)\right)ds\right] \\
		 & \leq 2\frac{\phi}{\rho}\,\mathbb{E}_\tau\left[\int_\tau^T |\delta P|\,|\sigma_1^2 - \sigma_2^2|\,ds\right]
	\end{align*}
	where we used $P_1, P_2 \geq 0$ so $(\delta P)^2 \rho(P_1+P_2)/c \geq 0$ and dropped the corresponding non-positive term.
	By the Cauchy--Schwarz inequality,
	\begin{equation*}
		\mathbb{E}_\tau\left[\int_\tau^T |\sigma_1^2 - \sigma_2^2|\,ds\right]
		\leq \sqrt{2}\left(\|\sigma_1\|_{\mathbb{H}_{BMO}}^2 + \|\sigma_2\|_{\mathbb{H}_{BMO}}^2\right)^{1/2}\|\sigma_1 - \sigma_2\|_{\mathbb{H}_{BMO}}.
	\end{equation*}
	Define $\Gamma := \sqrt{2}(\|\sigma_1\|_{\mathbb{H}_{BMO}}^2 + \|\sigma_2\|_{\mathbb{H}_{BMO}}^2)^{1/2}$, so that $g = 2(\phi/\rho)\Gamma$.
	Then for any stopping time $\tau$,
	\begin{equation*}
		|\delta P(\tau)|^2 \leq 2\frac{\phi}{\rho}\,\|\delta P\|_{\mathbb{S}^\infty}\,\Gamma\,\|\sigma_1 - \sigma_2\|_{\mathbb{H}_{BMO}}
		= g\,\|\delta P\|_{\mathbb{S}^\infty}\,\|\sigma_1 - \sigma_2\|_{\mathbb{H}_{BMO}}.
	\end{equation*}
	Taking the supremum over all stopping times $\tau \leq T$,
	\begin{equation*}
		\|\delta P\|_{\mathbb{S}^\infty}^2 \leq g\,\|\delta P\|_{\mathbb{S}^\infty}\,\|\sigma_1 - \sigma_2\|_{\mathbb{H}_{BMO}}.
	\end{equation*}
	Dividing by $\|\delta P\|_{\mathbb{S}^\infty}$ (or treating $\delta P \equiv 0$ trivially),
	\begin{equation*}
		\|\delta P\|_{\mathbb{S}^\infty} \leq g\,\|\sigma_1 - \sigma_2\|_{\mathbb{H}_{BMO}},
	\end{equation*}
	which gives $\|\delta P\|_{\mathbb{S}^\infty}^2 \leq g^2\|\sigma_1 - \sigma_2\|_{\mathbb{H}_{BMO}}^2$.

	\textit{Stability of $\Lambda$.}
	Set $\delta\Lambda = \Lambda_1 - \Lambda_2$ and $\delta Z^\Lambda = Z^{\Lambda,1} - Z^{\Lambda,2}$.
	The linearity of the driver in $\Lambda$ yields
	\begin{equation*}
		\delta\Lambda(\tau)
		= -\int_\tau^T \left(\delta\Lambda\left(\frac{\rho P_1}{c}+1\right) + \Lambda_2\frac{\rho\,\delta P}{c}\right)ds - \int_\tau^T \delta Z^\Lambda\,dX.
	\end{equation*}
	Applying It\^o's formula to $|\delta\Lambda(\tau)|^2$ yields an identity.
	Dropping the non-positive term $-2\int_\tau^T|\delta\Lambda|^2(\rho P_1/c+1)\,ds$ and the martingale integral, we obtain
	\begin{equation*}
		|\delta\Lambda(\tau)|^2 \leq 2\,\|\delta\Lambda\|_{\mathbb{S}^\infty}\,\left\|\frac{\rho}{c}\right\|_\infty C_\Lambda\,T\,\|\delta P\|_{\mathbb{S}^\infty}.
	\end{equation*}
	By the stability of $P$, $\|\delta P\|_{\mathbb{S}^\infty} \leq g\,\|\sigma_1-\sigma_2\|_{\mathbb{H}_{BMO}}$.
	Taking the supremum over $\tau$ and dividing by $\|\delta\Lambda\|_{\mathbb{S}^\infty}$,
	\begin{equation*}
		\|\delta\Lambda\|_{\mathbb{S}^\infty} \leq 2\,\left\|\frac{\rho}{c}\right\|_\infty C_\Lambda\,T\,g\,\|\sigma_1-\sigma_2\|_{\mathbb{H}_{BMO}},
	\end{equation*}
	so $\|\delta\Lambda\|_{\mathbb{S}^\infty}^2 \leq \tilde{C}\,g^2\|\sigma_1-\sigma_2\|_{\mathbb{H}_{BMO}}^2$ with $\tilde{C} = 4\|\rho/c\|_\infty^2 C_\Lambda^2 T^2$.

	\textit{Stability of $Q$.}
	From the ODE for $Q$, the explicit solution gives
	\begin{equation*}
		Q_i(t) = \rho\int_0^t e^{-\int_s^t \rho P_i/c\,du}\frac{\Lambda_i}{c}\,ds.
	\end{equation*}
	Using $|e^{-a}-e^{-b}| \leq |a-b|$ for $a,b \geq 0$ (by the mean value theorem) and $P_i \geq 0$,
	\begin{align*}
		|Q_1(t) - Q_2(t)|
		 & \leq \rho\int_0^t \left|\frac{\Lambda_1}{c} - \frac{\Lambda_2}{c}\right|\,ds + \rho\int_0^t \frac{|\Lambda_1|}{c}\int_s^t \frac{\rho|\delta P|}{c}\,du\,ds \\
		 & \leq \|\rho/c\|_\infty T\,\|\delta\Lambda\|_{\mathbb{S}^\infty} + \|\rho/c\|_\infty^2 C_\Lambda\,T^2\,\|\delta P\|_{\mathbb{S}^\infty}.
	\end{align*}
	Taking the square and using $(a+b)^2 \leq 2a^2+2b^2$,
	\begin{equation*}
		\|\delta Q\|_{\mathbb{S}^\infty}^2 \leq 2\|\rho/c\|_\infty^2 T^2\,\|\delta\Lambda\|_{\mathbb{S}^\infty}^2 + 2\|\rho/c\|_\infty^4 C_\Lambda^2 T^4\,\|\delta P\|_{\mathbb{S}^\infty}^2.
	\end{equation*}
	Substituting the stability estimates for $P$ and $\Lambda$ yields $\|\delta Q\|_{\mathbb{S}^\infty}^2 \leq \bar{C}\,g^2\|\sigma_1-\sigma_2\|_{\mathbb{H}_{BMO}}^2$.

	Finally, $g = 2^{3/2}(\phi/\rho)(\|\sigma_1\|_{\mathbb{H}_{BMO}}^2+\|\sigma_2\|_{\mathbb{H}_{BMO}}^2)^{1/2} \leq 4(\phi/\rho)R_\sigma$, so $g \to 0$ as $\phi \to 0$ on the fixed BMO ball.
\end{proof}

\subsection{Proof of Theorem \ref{thm:small_risk_aversion}}
Under the smallness condition of Theorem \ref{thm:small_risk_aversion}, we now prove existence and uniqueness near the benchmark $\phi=0$ equilibrium.
\begin{remark}
	The Riccati system is now coupled with $S$ through the quadratic term $\phi\sigma^2$ in the equation for $P$.
	Within the present fixed-point argument in $\mathbb{S}^\infty \times \mathbb{H}_{BMO}$, the condition
	\begin{equation*}
		\rho^2 T^2\|1/c\|_\infty^2 < 1
	\end{equation*}
	is structural.
	Indeed, the terminal term $\rho^2\hat Q(T)^2$ is controlled through
	\begin{equation*}
		\|\hat Q\|_{\mathbb{S}^\infty} \leq T\|1/c\|_\infty\|\hat\mu\|_{\mathbb{S}^\infty},
	\end{equation*}
	so the coefficient in front of $\|\hat\mu\|_{\mathbb{S}^\infty}^2$ becomes $\rho^2 T^2\|1/c\|_\infty^2$.
	The strict inequality is exactly what allows this term to be absorbed in the self-mapping estimate.
\end{remark}

\begin{proof}
	Fix $\tilde\phi > 0$; the final smallness threshold for $\phi$ will be chosen not to exceed $\tilde\phi$.

	\textit{Step 1: Riccati characterization and solution.}
	The derivation of the Riccati system follows the same ansatz as in the proof of Theorem \ref{thm:existence_equilibrium_no_risk_aversion}.
	Unlike the case $\phi = 0$, the equation for $P$ now depends on $\sigma$ through the coupling term $\phi\sigma^2$, so the Riccati system is no longer decoupled from the spot dynamics.
	By Lemma \ref{lemma:riccati_solution}, on every bounded BMO ball for $\sigma$ and for every fixed upper bound $\tilde\phi$ on $\phi$, the resulting Riccati system has a unique solution with $0 < \bar P \leq P \leq C_P(\tilde\phi,R_\sigma)$.
	The estimates used below require only the positivity of $P$ and the bounds $\|\Lambda\|_{\mathbb{S}^\infty} \leq C_\Lambda$ and $\|Q\|_{\mathbb{S}^\infty} \leq h_Q$, which remain independent of $\phi$ once the BMO ball is fixed.

	\textit{Step 2: Perturbation system.}
	Let $(\bar{S}, \bar{\sigma}, \bar{P}, \bar{\Lambda}, \bar{Q}, \bar{\mu})$ denote the $\phi = 0$ solution from Theorem \ref{thm:existence_equilibrium_no_risk_aversion}.
	From the equilibrium system \eqref{eq:equilibrium_system}, $\mu = c(m-q)$ and therefore $dQ = q\,dt = (m - \mu/c)\,dt$ with corresponding backward equation $d(cq) = - (\mu - \phi\sigma^2 Q)\,dt + Z\,dX$.
	For $\phi = 0$, $\bar{\mu} = c(m - \bar{q})$ satisfies $d(c\bar{q}) = -\bar{\mu}\,dt + \bar{Z}\,dX$.
	Setting $\hat{l} := l - \bar{l}$ for $l = S, \sigma, Q, \mu, Z$ and subtracting, we obtain $\hat{\mu} = -c\hat{q}$ and
	\begin{equation*}
		c(t)\hat{q}(t) = -\rho\hat{Q}(T) + \int_t^T (\hat{\mu} - \phi\sigma^2 Q)\,ds - \int_t^T \hat{Z}\,dX.
	\end{equation*}
	Replacing $\sigma = \hat{\sigma} + \bar{\sigma}$, $Q = \hat{Q} + \bar{Q}$, and $\hat{\mu} = -c\hat{q}$, the perturbation system is
	\begin{equation}\label{eq:perturbation_system}
		\begin{cases}
			\hat{S}(t)   & = \displaystyle - \int_t^T \hat{\mu}\,ds - \int_t^T \hat{\sigma}\,dX                                                                                                     \\
			\hat{Q}(t)   & = \displaystyle - \int_0^t \frac{\hat{\mu}}{c}\,ds                                                                                                                       \\
			\hat{\mu}(t) & = \displaystyle \rho\hat{Q}(T) - \int_t^T \left(\hat{\mu} - \phi\left(\hat{\sigma} + \bar{\sigma}\right)^2\left(\hat{Q} + \bar{Q}\right)\right)ds + \int_t^T \hat{Z}\,dX
		\end{cases}
	\end{equation}
	For $n \geq 1$, given $(\hat{S}^{n-1}, \hat{\sigma}_{n-1}) \in \mathbb{S}^\infty \times \mathbb{H}_{BMO}$, define $(\hat{Q}^n, \hat{\mu}^n)$ as the solution to
	\begin{equation}\label{eq:iteration}
		\begin{cases}
			\hat{Q}^n(t)   & = \displaystyle - \int_0^t \frac{\hat{\mu}^n}{c}\,ds                                                                                                                                   \\
			\hat{\mu}^n(t) & = \displaystyle \rho\hat{Q}^n(T) - \int_t^T \left(\hat{\mu}^n - \phi\left(\hat{\sigma}_{n-1} + \bar{\sigma}\right)^2\left(\hat{Q}^n + \bar{Q}\right)\right)ds + \int_t^T \hat{Z}^n\,dX
		\end{cases}
	\end{equation}
	which exists by the same Riccati construction as in Step 1, since $\hat{\sigma}_{n-1}$ is fixed and the iteration is performed on a bounded BMO ball.
	Then define $(\hat{S}^n, \hat{\sigma}_n)$ as the solution to
	\begin{equation*}
		\hat{S}^n(t) = - \int_t^T \hat{\mu}^n\,ds - \int_t^T \hat{\sigma}_n\,dX
	\end{equation*}
	with $\hat{\sigma}_0 := 0$.

	\textit{Step 3: Self-mapping and contraction.}
	Let $\mathcal{B}_R := \{(\hat{S}, \hat{\sigma}) : \|\hat{S}\|_{\mathbb{S}^\infty}^2 + \|\hat{\sigma}\|_{\mathbb{H}_{BMO}}^2 \leq R^2\}$ and assume that $(\hat{S}^{n-1}, \hat{\sigma}_{n-1})$ is in $\mathcal{B}_R$.
	For brevity, set
	\begin{equation*}
		\bar{R}^2 := \|\bar{\sigma}\|_{\mathbb{H}_{BMO}}^2 + R^2,
		\qquad
		\tilde R := R + \|\bar{\sigma}\|_{\mathbb{H}_{BMO}}.
	\end{equation*}
	Applying It\^o's formula to $|\hat{\mu}^n(\tau)|^2$ for any stopping time $\tau$,
	\begin{multline*}
		|\hat{\mu}^n(\tau)|^2 + \mathbb{E}_\tau\left[\int_\tau^T |\hat{Z}^n|^2\,ds\right]\\
		= \mathbb{E}_\tau\left[\rho^2\hat{Q}^n(T)^2 + 2\int_\tau^T \hat{\mu}^n\left(-\hat{\mu}^n + \phi\left(\hat{\sigma}_{n-1} + \bar{\sigma}\right)^2\left(\hat{Q}^n + \bar{Q}\right)\right)ds\right].
	\end{multline*}
	Dropping the non-positive terms $-2\int_\tau^T |\hat{\mu}^n|^2\,ds$ and $\mathbb{E}_\tau[\int_\tau^T |\hat{Z}^n|^2\,ds]$ from both sides gives
	\begin{equation*}
		|\hat{\mu}^n(\tau)|^2 \leq \rho^2\|\hat{Q}^n\|_{\mathbb{S}^\infty}^2 + 2\phi\,\|\hat{\mu}^n\|_{\mathbb{S}^\infty}\,\mathbb{E}_\tau\left[\int_\tau^T \left(\hat{\sigma}_{n-1} + \bar{\sigma}\right)^2\left|\hat{Q}^n + \bar{Q}\right|\,ds\right].
	\end{equation*}
	Using $\|\hat{Q}^n\|_{\mathbb{S}^\infty} \leq T\|1/c\|_\infty \|\hat{\mu}^n\|_{\mathbb{S}^\infty}$,
	\begin{equation*}
		|\hat{\mu}^n(\tau)|^2 \leq \rho^2T^2\|1/c\|_\infty^2\|\hat{\mu}^n\|_{\mathbb{S}^\infty}^2 + 2\phi\,\|\hat{\mu}^n\|_{\mathbb{S}^\infty}\,\mathbb{E}_\tau\left[\int_\tau^T \left(\hat{\sigma}_{n-1} + \bar{\sigma}\right)^2\left|\hat{Q}^n + \bar{Q}\right|\,ds\right].
	\end{equation*}
	From $\|\hat{Q}^n\|_{\mathbb{S}^\infty} \leq T\|1/c\|_\infty \|\hat{\mu}^n\|_{\mathbb{S}^\infty}$ and $\|\bar{Q}\|_{\mathbb{S}^\infty} \leq h_Q^0 := \|\rho/c\|_\infty\|\bar{\Lambda}\|_{\mathbb{S}^\infty} T$ from the $\phi = 0$ analysis, together with
	\begin{equation*}
		\mathbb{E}_\tau\left[\int_\tau^T (\hat{\sigma}_{n-1} + \bar{\sigma})^2\,ds\right] \leq 2\bar{R}^2,
	\end{equation*}
	it follows that
	\begin{equation*}
		\|\hat{\mu}^n\|_{\mathbb{S}^\infty}^2 \leq \rho^2T^2\|1/c\|_\infty^2\|\hat{\mu}^n\|_{\mathbb{S}^\infty}^2 + 4\phi\,\|\hat{\mu}^n\|_{\mathbb{S}^\infty}\left(T\|1/c\|_\infty \|\hat{\mu}^n\|_{\mathbb{S}^\infty} + h_Q^0\right)\bar{R}^2.
	\end{equation*}
	Define $\eta := 1-\rho^2T^2\|1/c\|_\infty^2$ and assume $\eta>0$.
	Then
	\begin{equation*}
		\eta\,\|\hat{\mu}^n\|_{\mathbb{S}^\infty}^2 \leq 4\phi\,\|\hat{\mu}^n\|_{\mathbb{S}^\infty}\left(T\|1/c\|_\infty \|\hat{\mu}^n\|_{\mathbb{S}^\infty} + h_Q^0\right)\bar{R}^2.
	\end{equation*}
	Dividing both sides by $\|\hat{\mu}^n\|_{\mathbb{S}^\infty}$ (or treating $\hat{\mu}^n \equiv 0$ trivially),
	\begin{equation*}
		\|\hat{\mu}^n\|_{\mathbb{S}^\infty}\left(\eta - 4\phi\,T\|1/c\|_\infty\bar{R}^2\right) \leq 4\phi\,h_Q^0\bar{R}^2.
	\end{equation*}
	For $\phi$ small enough so that $4\phi\,T\|1/c\|_\infty\bar{R}^2 \leq \eta/2$,
	\begin{equation*}
		\|\hat{\mu}^n\|_{\mathbb{S}^\infty}^2 \leq \frac{64}{\eta^2}\,\phi^2\left(h_Q^0\right)^2\bar{R}^4.
	\end{equation*}
	Applying It\^o's formula to $|\hat{S}^n|^2$ then yields
	\begin{equation*}
		\|\hat{S}^n\|_{\mathbb{S}^\infty}^2 + \|\hat{\sigma}_n\|_{\mathbb{H}_{BMO}}^2 \leq 4T^2\|\hat{\mu}^n\|_{\mathbb{S}^\infty}^2 \leq \frac{256}{\eta^2}\,\phi^2 T^2\left(h_Q^0\right)^2\bar{R}^4.
	\end{equation*}
	For this to be bounded by $R^2$, it suffices to take
	\begin{equation*}
		\phi \leq \frac{\eta R}{16\,T\,h_Q^0\bar{R}^2} =: \phi_1.
	\end{equation*}

	For the contraction, consider $(s^1, \Sigma_1)$ and $(s^2, \Sigma_2)$ in $\mathcal{B}_R$, and denote by $(\hat{S}^1, \hat{\sigma}_1)$ and $(\hat{S}^2, \hat{\sigma}_2)$ their respective images under the iteration, and by $Q^1, Q^2$ the corresponding full $Q$-processes.
	Applying It\^o's formula to $|\hat{S}^1(\tau) - \hat{S}^2(\tau)|^2$ for any stopping time $\tau$,
	\begin{equation*}
		|\hat{S}^1(\tau) - \hat{S}^2(\tau)|^2 + \mathbb{E}_\tau\left[\int_\tau^T |\hat{\sigma}_1 - \hat{\sigma}_2|^2\,ds\right]
		= \mathbb{E}_\tau\left[-2\int_\tau^T (\hat{S}^1 - \hat{S}^2)(\hat{\mu}^1 - \hat{\mu}^2)\,ds\right].
	\end{equation*}
	Using Cauchy--Schwarz and $2ab \leq \frac14 a^2 + 4b^2$,
	\begin{equation*}
		\|\hat{S}^1 - \hat{S}^2\|_{\mathbb{S}^\infty}^2 + \|\hat{\sigma}_1 - \hat{\sigma}_2\|_{\mathbb{H}_{BMO}}^2
		\leq \frac12\|\hat{S}^1 - \hat{S}^2\|_{\mathbb{S}^\infty}^2 + 8T^2\|\hat{\mu}^1 - \hat{\mu}^2\|_{\mathbb{S}^\infty}^2,
	\end{equation*}
	hence, after rearranging,
	\begin{equation}\label{eq:contraction_S}
		\|\hat{S}^1 - \hat{S}^2\|_{\mathbb{S}^\infty}^2 + \|\hat{\sigma}_1 - \hat{\sigma}_2\|_{\mathbb{H}_{BMO}}^2
		\leq 16T^2\|\hat{\mu}^1 - \hat{\mu}^2\|_{\mathbb{S}^\infty}^2.
	\end{equation}
	Applying It\^o's formula to $|\hat{\mu}^1(\tau) - \hat{\mu}^2(\tau)|^2$ and taking conditional expectation, dropping the non-positive term $-2\int_\tau^T|\hat{\mu}^1 - \hat{\mu}^2|^2\,ds$, we obtain
	\begin{multline*}
		|\hat{\mu}^1(\tau) - \hat{\mu}^2(\tau)|^2
		\leq \rho^2\|Q^1 - Q^2\|_{\mathbb{S}^\infty}^2 \\
		+ 2\phi\,\|\hat{\mu}^1 - \hat{\mu}^2\|_{\mathbb{S}^\infty}\,\mathbb{E}_\tau\left[\int_\tau^T \left|(\Sigma_1 + \bar{\sigma})^2 Q^1 - (\Sigma_2 + \bar{\sigma})^2 Q^2\right|\,ds\right].
	\end{multline*}
	By Lemma \ref{lemma:riccati_stability}, $\|Q^1 - Q^2\|_{\mathbb{S}^\infty} \leq \sqrt{\bar{C}}\,g\,\|\Sigma_1 - \Sigma_2\|_{\mathbb{H}_{BMO}}$ with $g \leq 4(\phi/\rho)\tilde{R}$.
	We decompose the integrand using $a^2 b - c^2 d = (a^2 - c^2)b + c^2(b - d)$:
	\begin{equation*}
		(\Sigma_1 + \bar{\sigma})^2 Q^1 - (\Sigma_2 + \bar{\sigma})^2 Q^2
		= \left[(\Sigma_1 + \bar{\sigma})^2 - (\Sigma_2 + \bar{\sigma})^2\right]Q^2 + (\Sigma_1 + \bar{\sigma})^2(Q^1 - Q^2).
	\end{equation*}
	For the first term, note that $(\Sigma_1 + \bar{\sigma})^2 - (\Sigma_2 + \bar{\sigma})^2 = (\Sigma_1 - \Sigma_2)(\Sigma_1 + \Sigma_2 + 2\bar{\sigma})$, so by the Cauchy--Schwarz inequality and $\|Q^2\|_{\mathbb{S}^\infty} \leq h_Q$,
	\begin{multline*}
		\mathbb{E}_\tau\left[\int_\tau^T \left|(\Sigma_1 - \Sigma_2)(\Sigma_1 + \Sigma_2 + 2\bar{\sigma})\right|\,|Q^2|\,ds\right] \\
		\leq 2h_Q\tilde{R}\,\|\Sigma_1 - \Sigma_2\|_{\mathbb{H}_{BMO}}.
	\end{multline*}
	where we used $\|\Sigma_1 + \Sigma_2 + 2\bar{\sigma}\|_{\mathbb{H}_{BMO}} \leq 2\tilde{R}$.
	For the second term, $\|(\Sigma_1 + \bar{\sigma})^2\|_{\mathbb{H}_{BMO}} \leq 4\tilde{R}^2$ and by the same stability estimate, $\|Q^1 - Q^2\|_{\mathbb{S}^\infty} \leq \sqrt{\bar{C}}\,g\,\|\Sigma_1 - \Sigma_2\|_{\mathbb{H}_{BMO}}$, hence
	\begin{equation*}
		\mathbb{E}_\tau\left[\int_\tau^T (\Sigma_1 + \bar{\sigma})^2 |Q^1 - Q^2|\,ds\right]
		\leq 16\sqrt{\bar{C}}\,\frac{\phi}{\rho}\,\tilde{R}^3\,\|\Sigma_1 - \Sigma_2\|_{\mathbb{H}_{BMO}}.
	\end{equation*}
	Combining both terms and taking the supremum over $\tau$,
	\begin{multline*}
		\|\hat{\mu}^1 - \hat{\mu}^2\|_{\mathbb{S}^\infty}^2
		\leq \underbrace{\rho^2\bar{C}\,g^2}_{=:\,B^2}\,\|\Sigma_1 - \Sigma_2\|_{\mathbb{H}_{BMO}}^2 \\
		+ 2\,\|\hat{\mu}^1 - \hat{\mu}^2\|_{\mathbb{S}^\infty}\,\underbrace{2\phi\,\tilde{R}\left(h_Q + 4\sqrt{\bar{C}}\,\frac{\phi}{\rho}\,\tilde{R}^2\right)}_{=:\,C}\,\|\Sigma_1 - \Sigma_2\|_{\mathbb{H}_{BMO}}.
	\end{multline*}
	Setting $A = \|\hat{\mu}^1 - \hat{\mu}^2\|_{\mathbb{S}^\infty}$, we have
	\begin{equation*}
		A^2 \leq B^2\|\Sigma_1-\Sigma_2\|^2 + 2AC\|\Sigma_1-\Sigma_2\|,
	\end{equation*}
	which implies
	\begin{equation*}
		(A - C\|\Sigma_1-\Sigma_2\|)^2 \leq (B^2+C^2)\|\Sigma_1-\Sigma_2\|^2,
	\end{equation*}
	and hence
	\begin{equation*}
		\|\hat{\mu}^1 - \hat{\mu}^2\|_{\mathbb{S}^\infty} \leq \left(B + 2C\right)\|\Sigma_1 - \Sigma_2\|_{\mathbb{H}_{BMO}} =: \kappa\,\|\Sigma_1 - \Sigma_2\|_{\mathbb{H}_{BMO}}.
	\end{equation*}
	Combined with \eqref{eq:contraction_S}, the iteration satisfies
	\begin{equation*}
		\|\hat{S}^1 - \hat{S}^2\|_{\mathbb{S}^\infty}^2 + \|\hat{\sigma}_1 - \hat{\sigma}_2\|_{\mathbb{H}_{BMO}}^2
		\leq 16T^2\kappa^2\|\Sigma_1 - \Sigma_2\|_{\mathbb{H}_{BMO}}^2.
	\end{equation*}
	Since $\kappa = O(\phi)$ is continuous in $\phi$ with $\kappa(0) = 0$, there exists $\phi_2 > 0$ such that $16T^2\kappa^2 < 1$ for all $\phi \in [0, \phi_2]$, making the iteration a contraction on $\mathcal{B}_R$.

	\textit{Step 4: Forward price determination.}
	For $\phi < \min(\tilde\phi,\phi_1, \phi_2)$, the Banach fixed-point theorem yields a unique fixed point $(\hat{S}^*, \hat{\sigma}^*) \in \mathcal{B}_R$.
	This determines $\hat{\mu}^*$, $\hat{Q}^*$ and hence the full solution $(S^*, \sigma^*, Q^*, \mu^*)$.
	Decomposing $Q^* = \hat{Q}^* + \bar{Q}$, the market clearing condition $E[Q^*(T)] = \mathfrak{s} - (F - E[\mathfrak{G}])/\rho$ with $\mathfrak{s} = \mathfrak{d}(F)$ reads
	\begin{equation*}
		\rho E[\alpha(T)]\,\mathfrak{d}(F) + E[\beta(T)] + E[\hat{Q}^*(T)] = \mathfrak{d}(F) - \frac{F - E[\mathfrak{G}]}{\rho}
	\end{equation*}
	where $\alpha, \beta$ are defined in the proof of Theorem \ref{thm:existence_equilibrium_no_risk_aversion} and satisfy $\rho E[\alpha(T)] = 1 - e^T P(0)$ in $(0, 1)$.
	From the self-mapping estimate,
	\begin{equation*}
		\|\hat{Q}^*\|_{\mathbb{S}^\infty} \leq T\|1/c\|_\infty\|\hat{\mu}^*\|_{\mathbb{S}^\infty} \leq \frac{8}{\eta}T\|1/c\|_\infty\phi\,h_Q^0\bar{R}^2
	\end{equation*}
	where $h_Q^0 = \|\rho/c\|_\infty\|\bar{\Lambda}\|_{\mathbb{S}^\infty}T$.
	In particular, for fixed data near the reference equilibrium $\bar{\mathfrak{s}} = \mathfrak{d}(\bar{F})$, we have $E[\hat{Q}^*(T)] = O(\phi)$.
	Define
	\begin{equation*}
		\Phi(F) := (\rho E[\alpha(T)] - 1)\mathfrak{d}(F) + (F - E[\mathfrak{G}])/\rho + E[\beta(T)] + E[\hat{Q}^*(\mathfrak{d}(F))(T)].
	\end{equation*}
	For $\phi = 0$, $\Phi_0(\bar{F}) = 0$ and $\Phi_0$ is strictly increasing because $(\rho E[\alpha(T)] - 1) < 0$, $\mathfrak{d}$ is strictly decreasing, and $(F-E[\mathfrak{G}])/\rho$ is strictly increasing.
	For $\phi > 0$ small enough, the fixed point depends continuously on the parameter $\mathfrak{s}$.
	Moreover, the contraction estimate implies that the map
	\begin{equation*}
		\mathfrak{s} \longmapsto E[\hat{Q}^*(\mathfrak{s})(T)]
	\end{equation*}
	is Lipschitz with some constant $L_\phi = O(\phi)$.
	Therefore, for $F_1 < F_2$,
	\begin{align*}
		\Phi(F_2) - \Phi(F_1)
		\geq {} & \left(1 - \rho E[\alpha(T)] - L_\phi\right)\left(\mathfrak{d}(F_1) - \mathfrak{d}(F_2)\right)
		+ \frac{F_2-F_1}{\rho}.
	\end{align*}
	Since $1 - \rho E[\alpha(T)] = e^T P(0) > 0$, we may choose $\phi$ small enough so that $L_\phi < 1 - \rho E[\alpha(T)]$.
	Then $\Phi$ is strictly increasing.
	Furthermore, as $F \to \pm\infty$, $\mathfrak{d}(F) \to \mp\infty$ and
	\begin{equation*}
		E[\hat{Q}^*(\mathfrak{d}(F))(T)] = O\left(\phi|\mathfrak{d}(F)|\right),
	\end{equation*}
	so the dominant part of $\Phi(F)$ remains the increasing deterministic contribution.
	Therefore $\Phi(F) \to \pm\infty$ as $F \to \pm\infty$.
	By the intermediate value theorem there exists a unique $F^*$ with $\Phi(F^*) = 0$.
	Choosing $\phi_0>0$ smaller than the thresholds above completes the proof.
\end{proof}

\bibliographystyle{abbrvnat}
\bibliography{biblio}

\end{document}